\documentclass[aps,pra,longbibliography,twocolumn,superscriptaddress]{revtex4-1}
\usepackage{bbm, amsmath, amssymb, amsthm, mathrsfs}
\usepackage{graphicx,epstopdf,color,verbatim,enumitem}
\usepackage[dvipsnames]{xcolor}
\usepackage{centernot}

\usepackage[normalem]{ulem} 

\definecolor{myurlcolor}{rgb}{0,0,0.7}
\definecolor{myrefcolor}{rgb}{0.8,0,0}
\usepackage[unicode=true,pdfusetitle, bookmarks=false,bookmarksnumbered=false,
bookmarksopen=false, breaklinks=false,pdfborder={0 0 0},backref=false,
colorlinks=true, linkcolor=myrefcolor,citecolor=myurlcolor,urlcolor=myurlcolor]{hyperref}

\newcommand{\1}{\openone}
\newcommand{\trm}[1]{\textrm{#1}}

\newcommand{\msf}[1]{\mathsf{#1}}
\newcommand{\op}[1]{\hat{#1}}

\newcommand{\ot}{\otimes}
\newcommand{\Ot}{\bigotimes}

\newcommand{\BE}{\begin{equation}}
\newcommand{\EE}{\end{equation}}
\newcommand{\BEA}{\begin{eqnarray}}
\newcommand{\EEA}{\end{eqnarray}}
\newcommand{\AL}[1]{\begin{align}#1\end{align}}

\def\dd{d}
\def\ii{\mathrm i}
\def\ee{e}

\def \cL{\mathcal L}
\def \cB{\mathcal B} 
\def \cH{\mathcal H}
\def \cI{\mathcal I}
\def \cD{\mathcal D}
\def \cZ{\mathcal Z}
\def \cX{\mathcal X}
\def \cU{\mathcal U}

\def \Cfun{\mathscr{C}}

\newcommand*\LIP[1][]{\bar{\cL}^{#1}}
\newcommand*\LSP[1][]{\cL^{#1}}
\newcommand*\Lenv[1][]{\tilde{\cL}^{#1}}
\newcommand*\LIPb[1][]{\bar{\mat{L}}^{#1}}
\newcommand*\LSPb[1][]{\mat{L}^{#1}}
\newcommand*\Lenvb[1][]{\tilde{\mat{L}}^{#1}}

\def\LambdaE{\tilde{\Lambda}}

\def\sV{{\boldsymbol \sigma}} 
\def\rV{\mathbf{r}}
\renewcommand{\vec}[1]{\mathbf{#1}}
\newcommand{\mat}[1]{\mathbf{\expandafter\MakeUppercase\expandafter{#1}}}

\def\sx{\sigma_x}
\def\sy{\sigma_y}
\def\sz{\sigma_z}
\def\bsz{{\bar\sigma}_z}
\def\sp{\sigma_+}
\def\sm{\sigma_-}
\def\hs{\sigma}

\def\rhoS{{\rho_S}}
\def\rhoSE{{\rho_{SE}}}
\def\rhoE{{\rho_{E}}}
\DeclareMathOperator{\tr}{Tr}

\newcommand{\eref}[1]{(\ref{#1})}
\newcommand{\eqnref}[1]{Eq.~\eref{#1}}
\newcommand{\eqnsref}[2]{Eqs.~\eref{#1}~and~\eref{#2}}
\newcommand{\figref}[1]{Fig.~\ref{#1}}

\newcommand{\secref}[1]{Sec.~\ref{#1}}
\newcommand{\appref}[1]{App.~\ref{#1}}

\newcommand{\defref}[1]{Def.~\ref{#1}}
\newcommand{\lemref}[1]{Lemma~\ref{#1}}

\newcommand{\refcite}[1]{Ref.~\cite{#1}}
\newcommand{\obsref}[1]{Obs.~\ref{#1}}

\newtheorem{thm}{Theorem}

\newtheorem{lem}[thm]{Lemma}
\newtheorem{definition}{Definition}
\newtheorem{obs}{Observation}

\DeclareGraphicsExtensions{.png,.pdf,.eps}
\graphicspath{ {./figs/} }

\begin{document}

\title{Adding dynamical generators in quantum master equations}

\author{Jan Ko\l{}ody\'{n}ski}
\affiliation{ICFO--Institut de Ciencies Fotoniques, The Barcelona Institute of Science and Technology, 08860 Castelldefels (Barcelona), Spain}
\author{Jonatan Bohr Brask}
\affiliation{Group of Applied Physics, University of Geneva, 1211 Geneva, Switzerland}
\author{Mart\'{i} Perarnau-Llobet}
\affiliation{ICFO--Institut de Ciencies Fotoniques, The Barcelona Institute of Science and Technology, 08860 Castelldefels (Barcelona), Spain}
\affiliation{Max-Planck-Institut f\"ur Quantenoptik, Hans-Kopfermann-Str.~1, D-85748 Garching, Germany}
\author{Bogna Bylicka}
\affiliation{ICFO--Institut de Ciencies Fotoniques, The Barcelona Institute of Science and Technology, 08860 Castelldefels (Barcelona), Spain}

\begin{abstract}
The \emph{quantum master equation} is a widespread approach to describing open quantum system dynamics. In this approach, the effect of the environment on the system evolution is entirely captured by the \emph{dynamical generator}, providing a compact and versatile description. However, care needs to be taken when several noise processes act simultaneously or the Hamiltonian evolution of the system is modified. Here, we show that generators can be added at the master equation level without compromising physicality only under restrictive conditions. Moreover, even when adding generators results in legitimate dynamics, this does not generally correspond to the true evolution of the system. We establish a general condition under which direct addition of dynamical generators is justified, showing that it is ensured under weak coupling and for settings where the free system Hamiltonian and all system-environment interactions commute. In all other cases, we demonstrate by counterexamples that the exact evolution derived microscopically cannot be guaranteed to coincide with the dynamics naively obtained by adding the generators.
\end{abstract}
\maketitle

\section{Introduction}
It is generally impossible to completely isolate a small system of interest from the surrounding environment. Thus, dissipative effects caused by the environment are important in almost every quantum experiment, ranging from highly controlled settings, where much effort is invested in minimising them, to areas where the dissipation is the key object of interest. In many cases, exact modelling of the environment is not practical and its effect is instead accounted for by employing effective models describing the induced noise. Different approaches exist, e.g., quantum Langevin and stochastic Schr\"odinger equations \cite{Gardiner,Breuer}, quantum jump and state-diffusion models \cite{Percival,Plenio1998}, or Hilbert-space averaging methods \cite{Gemmer2006}. 

Arguably, the most widely applied approach is to use the \emph{quantum master equation} (QME) description \cite{Gardiner,Breuer}. In this approach, the system evolution is given by a time-local differential equation, where the effect of the environment is captured by the \emph{dynamical generator}. A master equation can be derived from a microscopic model of the system and environment, and their interaction, by tracing over the environment and applying appropriate approximations \cite{Gardiner,Breuer}. However, QMEs are also often applied directly, without explicit reference to an underlying model. In that case, care needs to be taken when several noise processes act in parallel, as simultaneous coupling to multiple baths in a microscopic model does not generally correspond to simple addition of noise generators. Moreover, when the Hamiltonian evolution of the system is modified, e.g., when controlling system dynamics by coherent driving \cite{Schmidt2011}, the form of noise generators in a QME may significantly change. Additivity of noise at the QME level has been discussed recently for qubits when analysing dynamical effects of interference between different baths \cite{Chan2014,Mitchison2018}, non-additivity of relaxation rates in multipartite systems \cite{Yu2006,Lankinen2016}, as well as in the context of charge (excitation) transport \cite{Giusteri2016}.

In this work, we address the questions of when:
\begin{itemize}[leftmargin=16pt,noitemsep,topsep=2pt,partopsep=2pt,parsep=0pt]
 \item[(i)] \emph{The naive addition of generators yields physically valid dynamics.}
 \item[(ii)] \emph{The corresponding evolution coincides with the true system dynamics derived from the underlying microscopic model.}
\end{itemize}

First, we show that (i) is satisfied for generators which are commutative, semigroup-simulable (can be interpreted as a fictitious semigroup at each time instance), and preserve commutativity of the dynamics under addition. These reach beyond the case of Markovian generators for which (i) naturally holds. Outside of this class, we find examples of simple qubit QMEs which lead to unphysical dynamics. We observe that (ii) holds if and only if the cross-correlations between distinct environments can be ignored within a QME. We show this to be the case in the weak-coupling regime, extending previous results in this direction \cite{CohenTannoudji1998,Chan2014,Schaller2015}. We also provide a sufficient condition for (ii) dictated by the commutativity of Hamiltonians at the microscopic level. We combine these generic considerations with a detailed study of a specific open system, namely a qubit interacting simultaneously with multiple spin baths, for which we provide examples where (ii) is not satisfied, while choosing the microscopic Hamiltonians to fulfil particular commutation relations.

Our results are of relevance to areas of quantum physics where careful description of dissipative dynamics plays a key role, e.g., in dissipative quantum state engineering \cite{Diehl2008,Verstraete2009,Metelmann2015,Reiter2016}, dissipative coupling in optomechanics \cite{Aspelmeyer2014}, or in dissipation-enhanced quantum transport scenarios \cite{Gurvitz1996,Giusteri2016}, including biological processes \cite{Lambert2013}. In particular, they are of importance to situations in which QMEs are routinely employed to account for multiple sources of dissipation, e.g., in quantum thermodynamics \cite{Alicki1979,Levy2014arpc,Vinjanampathy2016,Goold2015} when dealing with multiple heat baths \cite{Skrzypczyk2011,Correa2013,Levy2014epl,Mitchison2018} or in quantum metrology \cite{Maccone2011,Escher2011,Demkowicz2012} where the relation between dissipation and Hamiltonian dynamics, encoding the estimated parameter, is crucial \cite{Chaves2013,Brask2015,Smirne2016,Haase2017}.

The manuscript is structured as follows. In \secref{sec:qmes}, we discuss QMEs at an abstract level---as defined by families of dynamical generators whose important properties we summarise in \secref{sec:dyn_gens_phys}. We specify  in \secref{sub:gen_add} conditions under which the addition of physically valid generators is guaranteed to yield legitimate dynamics. We demonstrate by explicit examples that even mild violation of these conditions may lead to unphysical evolutions.

In \secref{sec:micro}, we view the validity of QMEs from the microscopic perspective. In particular, we briefly review in \secref{sec:QME_micro_der} the canonical derivation of a QME based on an underlying microscopic model, in order to discuss the effect of changing the system Hamiltonian on the QME, as well as the generalisation to interactions with multiple environments. We then formulate a general criterion for the validity of generator addition in \secref{sec:gen_add_micro}, which we explicitly show to be ensured in the weak coupling regime, or when particular commutation relations of the microscopic Hamiltonians are fulfilled. 

In \secref{sec:magnets}, we develop an exactly solvable model of a qubit interacting with multiple spin baths, which allows us to explicitly construct counterexamples that disprove the microscopic validity of generator addition in all the regimes in which the aforementioned commutation relations do not hold. Finally, we conclude in \secref{sec:conclusion}.

\section{Time-local quantum master equations}
\label{sec:qmes}
QMEs constitute a standard tool to describe reduced dynamics of open quantum systems. They provide a compact way of defining the effective system evolution at the level of its density matrix, $\rho_S(t)$, without need for explicit specification neither of environmental interactions nor the nature of the noise. Although a QME may be expressed in a generalised form as an integro-differential equation involving time-convolution \cite{Vacchini2016}, its equivalent (c.f.~\cite{Chruscinski2010}) and more transparent \emph{time-local} formulation  is typically favoured, providing a more direct connection to the underlying physical mechanisms responsible for the dissipation \cite{Gardiner,Breuer}. Given a time-local QME:
\BE
\frac{d}{dt}\rho_S(t) = \cL_t[\rho_S(t)] = \cH_t[\rho_S(t)]+\cD_t[\rho_S(t)],
\label{eq:QME_dyn_gen}
\EE
all the information about the system evolution is contained within the \emph{dynamical generator}, $\cL_t$, that is uniquely defined at each moment of time $t$. Moreover, $\cL_t$ can always be decomposed into its \emph{Hamiltonian} and \emph{purely dissipative} parts, i.e., $\cL_t=\cH_t+\cD_t$ in \eqnref{eq:QME_dyn_gen} with $\cH_t[\rho] = -\ii [H(t) , \rho]$ and some Hermitian $H(t)$ \cite{Gorini1976}. 

Although the QME \eref{eq:QME_dyn_gen} constitutes an ordinary differential equation, the system evolution may exhibit highly non-trivial memory features thanks to the arbitrary dependence of $\cL_t$ on the local time-instance $t$, but also on the (fixed) initial time $t_0$ at which the evolution commences~\cite{Chruscinski2010}---which, without loss of generality, we choose to be zero ($t_0=0$) and drop throughout this work.

\subsection{Physicality of dynamical generators}
\label{sec:dyn_gens_phys}
For the QME \eref{eq:QME_dyn_gen} to be physically valid, it must yield dynamics that is consistent with quantum theory.
In particular, upon integration the QME must lead to a \emph{family of (dynamical) maps} $\Lambda_t$ (parametrised by $t$) that satisfy $\rhoS(t) = \Lambda_t[\rhoS(0)]$ for any $t\ge 0$ and initial $\rhoS(0)$, with each $\Lambda_t$ being completely positive and trace preserving (CPTP) \cite{Jamiolkowski1972,Choi1975}.

On the other hand, any QME \eref{eq:QME_dyn_gen} is unambiguously specified by the \emph{family of (dynamical) generators} $\cL_t$ appearing in \eqnref{eq:QME_dyn_gen}. However, as discussed in \appref{app:QMEs_dyn_gen_fams}, although the CPTP condition can be straightforwardly checked for maps $\Lambda_t$, it does not directly translate onto the generators $\cL_t$. As a result, for a generic QME its physicality cannot be easily inferred at the level of \eqnref{eq:QME_dyn_gen}, unless its explicit integration is possible. Nevertheless, we formally call a family of dynamical generators $\cL_t$ \emph{physical} if the family of maps it generates consists only of CPTP transformations. In what follows (see also \appref{app:dyn_gen_descr}), we describe properties of dynamical generators that ensure their physicality.

Any family of dynamical generators, whether physical or not, can be uniquely decomposed as \cite{Gorini1976}
\BE
\cL_t[\rho]=
-\ii [H(t) , \rho]+\sum_{i,j=1}^{d^2-1} \msf{D}_{ij}(t) \left( F_i \rho F_j^\dagger - \frac{1}{2}\{ F_j^\dagger F_i , \rho \} \right)\!,
\label{eq:gen_mat_basis}
\EE
where $d$ is the Hilbert space dimension and $\{F_i\}_{i=1}^{d^2}$ is any orthonormal operator basis with $\tr\{F_i^\dagger F_j\}=\delta_{ij}$ and all $F_i$ traceless except $F_{d^2}=\openone/\sqrt{d}$. The Hamiltonian part, $\cH_t$, of the generator in \eqnref{eq:QME_dyn_gen} is then determined by $H(t)$ of \eqnref{eq:gen_mat_basis}, while the dissipative part $\cD_t$ is defined by the Hermitian matrix $\msf{D}(t)$. Although general criteria for physicality of dynamical-generator families are not known, two natural classes of physical dynamics can be identified based on the above decomposition.

In particular, when $\msf{D}(t)$ is positive semidefinite, $\cL_t$ is said to be of Gorini-Kossakowski-Sudarshan-Lindblad (GKSL) form \cite{Gorini1976,Lindblad1976}. If this is the case for all $t\ge0$, then the corresponding evolution is not only physical but also \emph{CP-divisible}, i.e., the corresponding family of maps can be decomposed as $\Lambda_{t}=\tilde{\Lambda}_{t,s}\Lambda_{s}$, where $\tilde{\Lambda}_{t,s}$  is CPTP for all $0\le s\le t$. This property is typically associated with Markovianity of the evolution \cite{Rivas2014,Breuer2016,deVega2017}. 

Furthermore, when, in addition, $H$ and $\msf{D}$ in \eqnref{eq:gen_mat_basis} are time-independent, the dynamics forms a \emph{semigroup}, such that the generator and map families are directly related via $\Lambda_t=\exp[t\cL]$ with all $\cL_t=\cL$ \cite{Alicki2002}. See \appref{app:dyn_gen_descr} for a more detailed discussion of different types of evolutions.

For the purpose of this work, we also identify another important class of physical dynamics:
\begin{definition}
A given dynamical family $\Lambda_t$ is semigroup-simulable (SS) if for any $t\ge0$ the map $\cZ_t =\log \Lambda_t$ is of the GKSL form \eref{eq:gen_mat_basis} with some $\msf{D}(t)\ge0$. 
\end{definition}
\noindent Formally, $\cZ_t$ constitutes the \emph{instantaneous exponent} of the dynamics, satisfying $\Lambda_t=\ee^{\cZ_t}$ (see \refcite{Chruscinski2014} and \appref{app:dyn_gen_descr}). Physicality of the evolution is then guaranteed by the GKLS form of $\cZ_t$, because at any $t$ the dynamical map $\Lambda_t$ can be interpreted as a fictitious semigroup $\Lambda_t=\left.\ee^{\cZ_t\tau}\right|_{\tau=1}$ generated by $\cZ_t$ (at this particular time instance). $\Lambda_t$ must therefore be CPTP at $t$.

In general, it is not straightforward to verify whether a given QME \eref{eq:QME_dyn_gen} yields SS dynamics \cite{Chruscinski2014}, even after decomposing its dynamical generators according to \eqnref{eq:gen_mat_basis}. However, in the special case of \emph{commutative} dynamics, for which $[\cL_s,\cL_t]=0$ (or equivalently $[\Lambda_s,\Lambda_t]=0$) for all $s,t\ge0$, one may directly identify the SS subclass, because (see \appref{app:dyn_gen_descr}):
\begin{lem}
\label{lem:ss_cond_comm}
Any commutative dynamics is SS iff (if and only if) for any $t\ge0$ the decomposition of its dynamical generators \eref{eq:gen_mat_basis} fulfills
\BE
\int_{0}^{t}\!{\dd}\tau\,\msf{D}(\tau)\ge0.
\label{eq:ss_cond_comm}
\EE
In short, we term any such semigroup-simulable and commutative dynamics SSC.
\end{lem}
\noindent Note that the condition \eref{eq:ss_cond_comm} is clearly weaker than positive semi-definiteness, $\msf{D}(t)\ge0$, at all times. Hence, there exist commutative dynamics which are SS but not CP-divisible. However, let us emphasise that there also exist commutative dynamics which are physical but \emph{not} even SS. An explicit example is provided by the eternally non-Markovian model of \refcite{Hall2014}, as well as by other instances of random unitary \cite{Andersson2007,Chruscinski2013}  and phase covariant \cite{Smirne2016} qubit dynamics, which we discuss in detail in \appref{app:dyn_gens_qubit_dyns}.
\begin{figure}[t]
\begin{center}
\includegraphics[width=\linewidth]{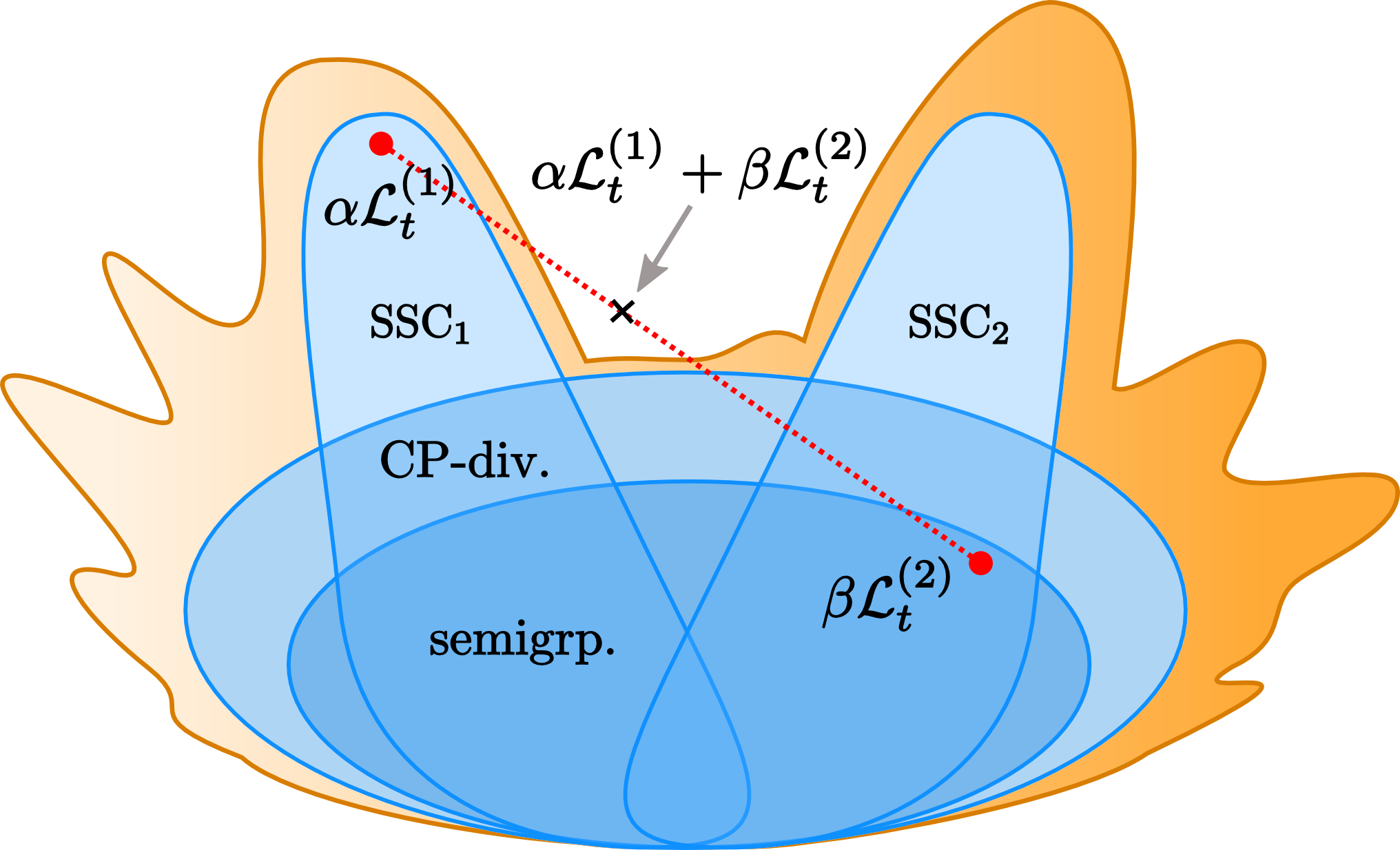}
\end{center}
\caption{\textbf{Cross-section of the vector space defined by the families of dynamical generators.} The non-convex set (\emph{orange}) describes a cut through the set of all physical families, while the inner convex sets (\emph{blue}) correspond to cuts through convex cones of various dynamical subclasses possessing additive generators. Sets containing CP-divisible and semigroup evolutions are indicated, as well as two exemplary SSC classes of dynamics. Physicality can be broken by adding to a family $\cL^{(1)}_t$, which is SSC but \emph{non}-CP-divisible, another family $\cL^{(2)}_t$ that lies outside of the particular SSC class, even a semigroup.}
\label{fig:generatorgeometry}
\end{figure}
%

\subsection{Additivity of dynamical generators}
\label{sub:gen_add}
We define the notion of \emph{additivity} for families of dynamical generators as follows:
\begin{definition}
Two physical families of generators $\cL^{(1)}_t$ and $\cL^{(2)}_t$ are additive if all their non-negative linear combinations, $\alpha \cL^{(1)}_t+\beta \cL^{(2)}_t$ with $\alpha,\beta\ge0$, are also physical.
\label{def:gen_additivity}
\end{definition}
\noindent Note that according to \defref{def:gen_additivity} a pair of generator families can be additive only if each of them is individually \emph{rescalable}---remains physical when multiplied by an non-negative scalar, i.e., $\cL_t\to\alpha\cL_t$ remains physical for any $\alpha\ge0$. However, as such a multiplication does not invalidate the GKSL form of the decomposition \eref{eq:gen_mat_basis} or the condition \eref{eq:ss_cond_comm}, it follows that any generator family which is CP-divisible or SSC must be rescalable. 

From the linear algebra perspective \cite{Rockafellar1970}, one may formally define the \emph{vector space} containing families of dynamical generators. Physical generators then form its particular subset. Rescalability of a given $\cL_t$ states then that the whole \emph{ray} $\{\alpha\cL_t\}_{\alpha\ge0}$ lies within the physical set. Additivity of $\cL^{(1)}_t$ and $\cL^{(2)}_t$, on the other hand, means that all the elements of the \emph{convex cone}, $\{\alpha \cL^{(1)}_t+\beta \cL^{(2)}_t\}_{\alpha,\beta\ge0}$, are physical.

For CP-divisible dynamics, we observe that when both $\cL^{(1)}_t$ and $\cL^{(2)}_t$ are of the GKLS form or even form a semigroup, so must any non-negative linear combination of them. Hence, it naturally follows that generator families describing CP-divisible evolutions constitute a convex cone contained in the physical set, with semigroups forming a subcone. Furthermore, as we demonstrate in \appref{app:dyn_gens_add}:
\begin{lem}
Any SSC generator families $\cL^{(1)}_t$ and $\cL^{(2)}_t$ are additive if upon addition, $\alpha \cL^{(1)}_t+\beta \cL^{(2)}_t$ with any $\alpha,\beta\ge0$, they yield commutative dynamics.
\end{lem}
\noindent Hence, the non-negative linear span of any such SSC pair forms a convex cone contained in the physical set. Moreover, one may then naturally expand such a cone by considering more than two, in particular, a complete set of SSC generator families whose non-negative linear combinations are all commutative. We term the convex cone so-constructed a particular \emph{SSC class}.

In \figref{fig:generatorgeometry}, we schematically depict the cross-section of the set of physical generator families, which then also cuts through the convex cones containing generator families of the aforementioned dynamical subclasses. Importantly, as all physical dynamics do \emph{not} form a convex cone in the vector space, the ones lying within such a hyperplane are described by a \emph{non-convex} set that, in turn, contains the \emph{convex} sets of:~CP-divisible dynamics, its semigroup subset, as well as ones representing particular SSC classes.

Now, as indicated in \figref{fig:generatorgeometry} by the dashed line, by adding a generator family that is SSC but \emph{not} CP-divisible (i.e., non-Markovian \cite{Rivas2014,Breuer2016,deVega2017}) and another physical family, even a semigroup, which does not commute with the first---i.e., is not contained within the corresponding SSC class---one may obtain unphysical dynamics. Consider an example of two purely dissipative ($\cL_t=\cD_t$ in \eqnref{eq:QME_dyn_gen}) qubit generators:
\begin{subequations}
\label{eq:ex_gens}
\begin{align}
\cL^{(1)}_t[\rho] & = \gamma_1(t)\, (\sx \rho \sx - \rho),  \label{eq:gen_deph}\\ 
\cL^{(2)}_t[\rho] & = \gamma_2(t)\, (\sm \rho \sp - \frac{1}{2}\{\sp\sm,\rho\}), \label{eq:gen_emiss}
\end{align}
\end{subequations}
where $\hs_{\pm} = (\sx \pm \ii\sy)/2$ and $\sx$, $\sy$, $\sz$ are the Pauli operators, and $\gamma_1(t)$, $\gamma_2(t)$ are chosen such that the generators are physical. Importantly, the families $\cL^{(1)}_t$ and $\cL^{(2)}_t$ despite being commutative, do not commute between one another. They belong to different SSC classes of qubit dynamics (see \appref{app:dyn_gens_qubit_dyns}), namely, random-unitary \cite{Andersson2007,Chruscinski2013} and phase-covariant \cite{Smirne2016} evolutions, respectively. 

In order to prove the situation indicated in \figref{fig:generatorgeometry}, we construct examples in which both $\cL^{(1)}_t$ and $\cL^{(2)}_t$ are physical, but their sum is not. We take instances of  $\gamma_1(t)$ and $\gamma_2(t)$ with one rate being constant (semigroup), and the other taking negative values for some times (non-Markovian) while fulfilling $\int_0^t\dd\tau\gamma(\tau)\ge0$ of \eqnref{eq:ss_cond_comm} (SSC). Two simple examples are provided by choosing $\gamma_1(t) = \sin(\omega t)$ and $\gamma_2(t) = \gamma$ and \emph{vice versa}, with $\gamma$ and $\omega$ being positive constants. We consider then the generator family  $\cL^{(1)}_t + \cL^{(2)}_t$ and solve analytically  in \appref{app:counter_gens_add} for the families of maps $\Lambda_t$ that arise in both cases. For each $\Lambda_t$, we compute the eigenvalues of its Choi-Jamio\l{}kowski (CJ) matrix---all of which must be non-negative at all times for the map to be CPTP (see \appref{app:dyn_gen_descr}).  We depict them as a function of time in \figref{fig:CJevals} for a choice of parameters which clearly demonstrates that the physicality is, indeed, invalidated at finite times. Their negativity, as demonstrated in \appref{app:counter_gens_add}, can also be verified analytically.

In \appref{app:counter_gens_add}, we also consider additional choices of $\gamma_1(t)$ and $\gamma_2(t)$ for the generators \eref{eq:ex_gens}, in order to show that the same conclusion holds when both semigroup and non-Markovian contributions come from explicit microscopic derivations. In particular, as the generators describe dephasing \eref{eq:gen_deph} and spontaneous-emission \eref{eq:gen_emiss} processes, we consider their non-Markovian forms derived (see \appref{app:qubit_classes}) from spin-boson and Jaynes-Cummings models, respectively \cite{Breuer}. 

Note that it follows from the above observations that physicality of a non-Markovian QME can be easily broken by addition of even a time-invariant (semigroup) dissipative term. Moreover, one should be extremely careful when dealing with dynamics described by generator families that are not even rescalable, e.g., see \appref{app:dyn_gens_rescal}:~ones that exhibit singularities at finite times \cite{Andersson2007}, are derived assuming weak-coupling interactions \cite{Gaspard1999}, or lead to physical (even commutative) but non-SS dynamics \cite{Hall2014}.
\begin{figure}[t!]
\includegraphics[width=0.8\linewidth]{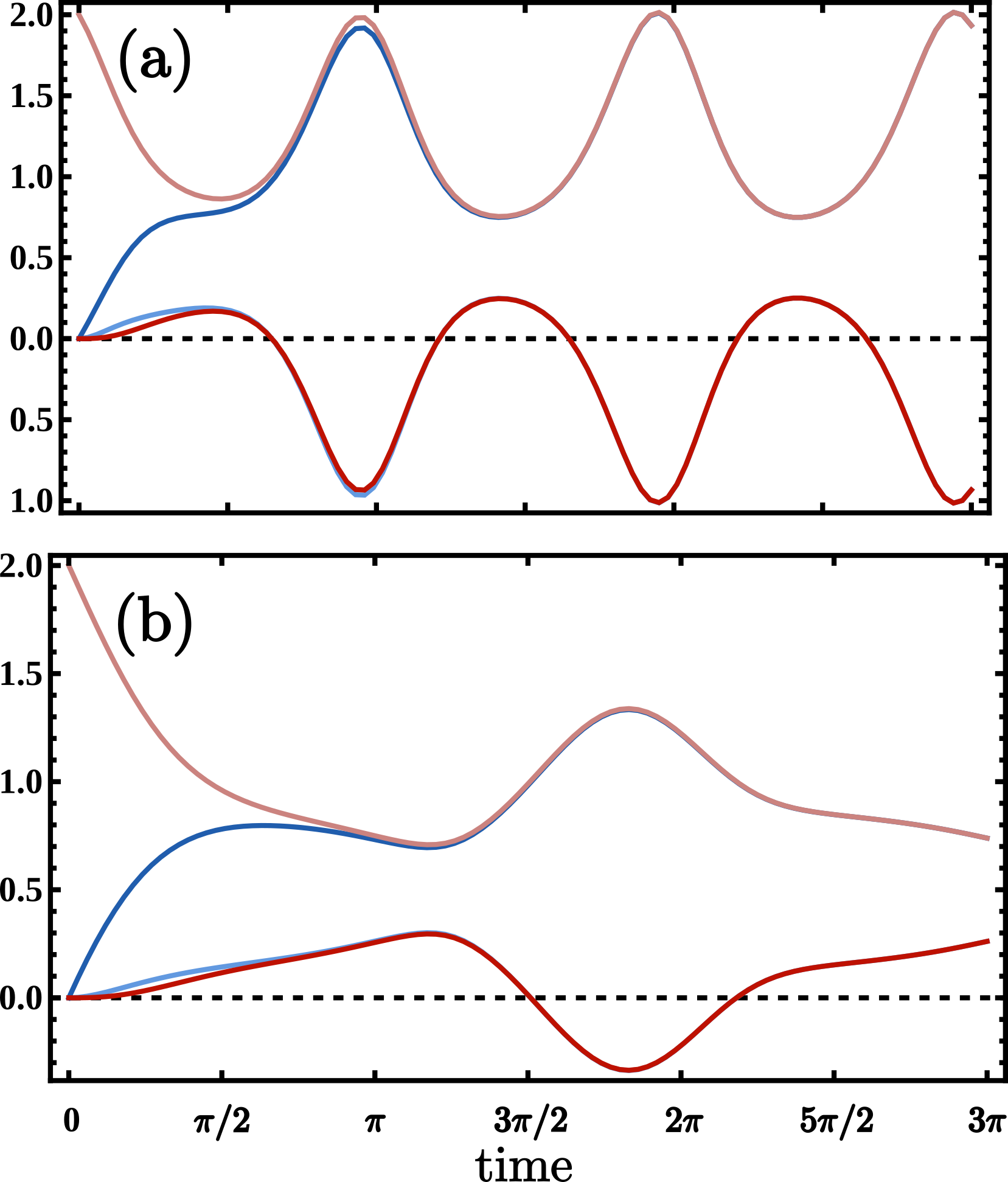}
\caption{%
\textbf{Eigenvalues of the CJ matrix as a function of time}, whose negativity demonstrates non-physicality of the dynamical maps generated by $\cL^{(1)}_t + \cL^{(2)}_t$ with $\cL^{(\bullet)}_t$ as defined in \eqnref{eq:ex_gens}. We choose in (a):~$\gamma_1(t) = \sin(2 t)$ and $\gamma_2(t) = 1$;~while in (b):~$\gamma_1(t) = 1/2$ and $\gamma_2(t) = \sin(t)$.}
\label{fig:CJevals}
\end{figure}
%

\section{Microscopic approach to QME\MakeLowercase{s}}
\label{sec:micro}
Let us recall that the QME \eref{eq:QME_dyn_gen} constitutes an effective description of the reduced dynamics, whose form must always originate from an underlying physical mechanism responsible for both free (noiseless) and dissipative parts of the system evolution. In particular, given a \emph{microscopic model} one should arrive at \eqnref{eq:QME_dyn_gen} starting from a closed dynamics describing the evolution of:~the system, its environment, as well as their interaction;~after tracing out the environmental degrees of freedom \cite{Gardiner,Breuer}.

\subsection{Microscopic derivation of a QME}
\label{sec:QME_micro_der}
In a microscopic model of an evolving open quantum system, as illustrated in \figref{fig.settings}(a), one considers a system of interest $S$, coupled to an environment $E$ that is taken sufficiently large for the total system to be closed. The global evolution is then unitary, $U_{SE}(t)=\exp[-\ii(H_S + H_E + H_I)t]$, being determined by the free Hamiltonians $H_S$ and $H_E$, and the system-environment interaction $H_I$. In the \emph{Schr\"odinger picture}, the reduced state of the system, $\rhoS(t)=\tr_E \rhoSE(t)$, evolves as
\BE
\frac{d}{dt}\rhoS(t) = - \ii \tr_E \left[H_S + H_E + H_I , \rhoSE(t)\right],
\label{eq:ODE_S+E}
\EE
where $\rhoSE$ is the total system-environment state. 

If the environment and the system are initially uncorrelated, so that $\rhoSE(0)=\rhoS(0)\otimes\rhoE$, and $\rhoE$ is stationary, i.e., $\left[H_{E},\rhoE\right]=0$, \eqnref{eq:ODE_S+E} can be conveniently rewritten as (see also \appref{app:QME_integrodiff}) \cite{Breuer,Rivas2012}:
\BE
\frac{d}{dt}\bar{\rho}_S(t) = - \int_0^t ds \tr_E [ \bar{H}_{I}(t) , [ \bar{H}_{I}(s) , \bar{\rho}_{SE}(s) ]],
\label{eq:exact_integrodiff_singlebath}
\EE
where by the bar, $\bar{\bullet} :=\ee^{\ii\left(H_{S}+H_{E}\right)t}\bullet\,\ee^{-\ii\left(H_{S}+H_{E}\right)t}$, we denote the \emph{interaction picture} with respect to the free system-environment Hamiltonian $H_S + H_E$.
\eqnref{eq:exact_integrodiff_singlebath} constitutes the integro-differential QME discussed at the beginning of \secref{sec:qmes} that, in practice, is typically recast into the time-local form \eref{eq:QME_dyn_gen}, which after returning to the Schr\"{o}dinger picture (see \appref{app:QME_TL}~and~\ref{app:QME_SP}) reads: 
\BE
\frac{d}{dt}\rho_S(t) = -\ii[H_S,\rho_S(t)]+\Lenv_t\!\left[\rho_S(t)\right].
\label{eq:QME}
\EE
Importantly, $\Lenv_t$ above can be unambiguously identified as the dynamical generator---containing both Hamiltonian and dissipative parts as in \eqnref{eq:QME_dyn_gen}---that arises purely due to the interaction with the environment;~with the system \emph{free} evolution (dictated by the system Hamiltonian $H_S$) being explicitly separated.
\begin{figure}[!t]
\includegraphics[width=\columnwidth]{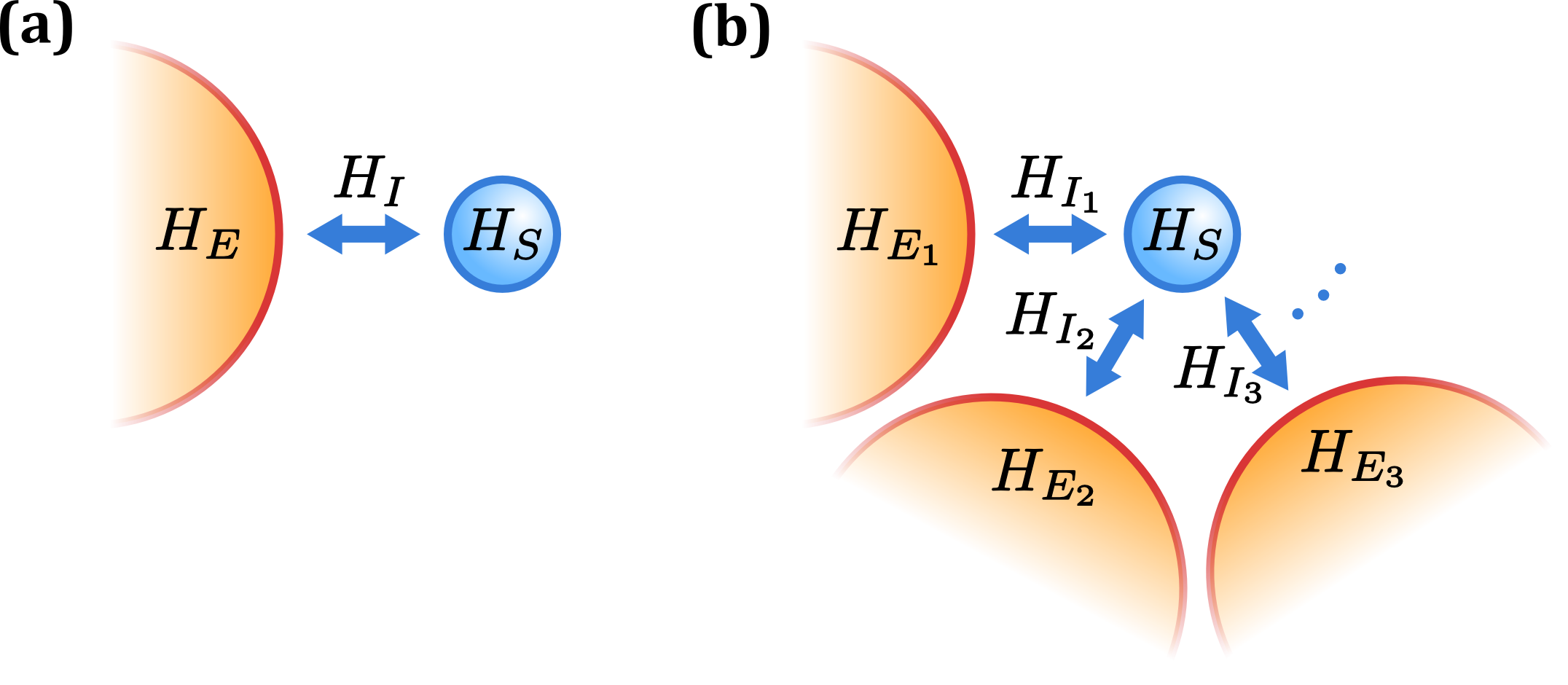}
\caption{%
\textbf{Microscopic description of an open quantum system $S$},
(a):~interacting with a single environment $E$; (b):~simultaneously interacting with multiple, independent environments $E_1$, $E_2$, $E_3$, $\dots$.}
\label{fig.settings}
\end{figure}
%

\subsubsection{Dependence on the system Hamiltonian}
However, as detailed in \appref{app:HS_cov}, despite the separation of terms in \eqnref{eq:QME} the form of the dynamical generator, $\Lenv_t$, may in general strongly depend on the system Hamiltonian $H_S$. Crucially, this means that the evolution of systems with different $H_S$, which interact with the same type of environment, cannot generally be modelled with the same QME after simply changing the $H_S$ in \eqnref{eq:QME}. However, under certain circumstances this can be justified.

In \appref{app:QME_HS}, we discuss in detail the natural cases when variations of $H_S$ do not affect the form of $\Lenv_t$ in \eqnref{eq:QME}, yet we summarise them here by the following lemma:
~\\
\begin{lem}
Consider a change $H_S \rightarrow H^{\prime}_S(t) = H_S + V(t)$ in \eqnref{eq:ODE_S+E}. The corresponding time-local QME can be obtained by just replacing $H_S$ with $H^{\prime}_S(t)$ in \eqnref{eq:QME} while keeping $\Lenv_t$ unchanged, if at all times $[V(t),H_I] = 0$ and either $[H_S, H_I] = 0$ or $[V(t), H_S] = 0$ (or both).
\end{lem}
%
Unfortunately, if the above sufficient condition cannot be met, one must, in principle, rederive the QME \eref{eq:QME} and the corresponding generator $\Lenv_t$ for $H^{\prime}_S(t)$. Moreover, such treatment is required independently of the interaction strength, i.e., also in the weak-coupling regime discussed below. A prominent physical example is provided by coherently driven systems, for which $V(t)$ represents the externally applied force. In their case, it is common that the time-dependence of $V(t)$ is naturally carried over onto, and significantly amends, the dynamical generator irrespectively of the coupling strength \cite{Rivas2010,Rivas2012}.

\subsubsection{Generalisation to multiple environments}
Another important question one should pose is under what conditions the full derivation of the QME \eref{eq:QME} can also be bypassed when dealing with a system that simultaneously interacts with multiple environments---as depicted in \figref{fig.settings}(b). Motivated by the analysis of \secref{sub:gen_add}, one may then naively expect that, given multiple \emph{additive} generator families describing each separate interaction, $\Lenv[(i)]_t$, they should be simply added to construct the overall QME of the form \eref{eq:QME} with $\Lenv_t=\sum_i\Lenv[(i)]_t$ \footnote{Note that we are not concerned with the internal structure of the system that is crucial when discussing, e.g., additivity of decay rates for a bipartite system with each of its parts coupled to a different reservoir \cite{Yu2006}.}.

Such a procedure may, however, lead to incorrect dynamics, as may be demonstrated by considering explicitly the microscopic model that incorporates interactions with multiple environments---with now $H_E = \sum_i H_{E_i}$ and $H_I = \sum_i H_{I_i}$ in \eqnref{eq:ODE_S+E}. Following the derivation steps of the time-local QME \eref{eq:QME}, while assuming its existence both in the presence of each single environment and all of them, one arrives at a generalised QME (see also \appref{app:valid_add_gens}):
\begin{widetext}
\BE
\frac{d}{dt}\rhoS(t) =-\ii\left[H_{S},\rhoS(t)\right]+\sum_{i}\Lenv[(i)]_{t}\!\left[\rhoS(t)\right]
-\sum_{i\neq j}\int_{0}^{t}\!\!ds\,\ee^{-\ii H_{S}(t-s)}\tr_{E_{ij}}\left[\bar{H}_{I_{i}}(t-s),\left[H_{I_{j}},\rho_{SE_{ij}}(s)\right]\right]\ee^{\ii H_{S}(t-s)},
\label{eq:QME_multiple_envs}
\EE
\end{widetext}
where $\bar{H}_{I_{i}}(\tau)=\ee^{\ii\left(H_{S}+H_{E_{i}}\right)\tau}H_{I_{i}}\ee^{-\ii\left(H_{S}+H_{E_{i}}\right)\tau}$, $\Lenv[(i)]_t$ is the generator arising when only the $i$th environment is present, $\rho_{SE_{ij}}$ denotes the joint-reduced state of the system and environments $i$ and $j$, while $\tr_{E_{ij}}$ stands for the trace over these environments. 

Crucially, the naive addition of generators would lead to a QME that contains only the first two terms in \eqnref{eq:QME_multiple_envs}. In particular, it would completely ignore the last term, which we here name the \emph{cross-term}, as it accounts for the cross-correlations that may emerge between each two environments due to their indirect interaction being mediated by the system.

\subsection{Microscopic validity of generator addition}
\label{sec:gen_add_micro}
The generalisation of the QME to multiple environments \eref{eq:QME_multiple_envs} allows one to unambiguously identify when the true dynamics derived microscopically coincides with the evolution obtained by naively adding the generators.
\begin{obs}
A dynamical generator corresponding to a system simultaneously interacting with multiple environments can be constructed by simple addition of the generators associated with each individual environment iff the cross-term in \eqnref{eq:QME_multiple_envs} identically vanishes. 
\label{obs:crossterm_vanish}
\end{obs}
\noindent %
In what follows, we show that \obsref{obs:crossterm_vanish} allows one to prove the validity of generator addition in the weak-coupling regime. However, as the cross-term in \eqnref{eq:QME_multiple_envs} involves a time-convolution integral, in order to prove that it identically vanishes given the microscopic Hamiltonians satisfy particular commutation relations, we must consider the dynamics in its integrated form---at the level of the corresponding dynamical (CPTP) map.

\subsubsection{Weak-coupling regime}
\label{app:valid_add_gens_weakcoupling}
\begin{lem}
The dynamical generator of the evolution of a system simultaneously interacting with multiple environments in the weak coupling regime can be constructed by simple addition of the generators associated with each individual environment. 
\label{lem:weak_coupling}
\end{lem}
\noindent %
Here, we summarise the proof of \lemref{lem:weak_coupling} that can be found in \appref{app:valid_add_gens_weakcoupling}, and generalises the argumentation of \refcite{CohenTannoudji1998,Schaller2015} applicable to the more restricted regime in which the Born-Markov approximation is valid. 

In particular, it applies to any QME valid in the \emph{weak-coupling} regime (c.f.~\cite{deVega2017,Breuer,Rivas2012}) which is derived using the ansatz:
\BE
\rhoSE(t) \approx  \rhoS(t) \otimes \bigotimes_i  \varrho_{E_i}(t),
\label{eq:weakcopling_appr}
\EE
where $\rhoS(t)$ is the reduced system state at $t$, while $\varrho_{E_i}(t)$ can be arbitrarily chosen for $t>0$---it does not need to represent the reduced state of the $i$th environment, $\rho_{E_i}(t)=\tr_{\neg E_i}\rhoSE(t)$, as long as it initially coincides with its stationary state, i.e., $\varrho_{E_i}(0)=\rho_{E_i}$. Let us emphasise that the assumption \eref{eq:weakcopling_appr} is employed only at the derivation stage, so that the QME so-obtained, despite correctly reproducing the reduced dynamics of the system under weak coupling, may yield upon integration closed dynamics with overall system-environment states that significantly deviate from the ansatz \eref{eq:weakcopling_appr} and, in particular, its tensor-product structure \cite{Rivas2010}.

In \appref{app:valid_add_gens_weakcoupling}, we first employ operator Schmidt decomposition \cite{Bengtsson2006} to reexpress each of the interaction Hamiltonians in \eqnref{eq:QME_multiple_envs} as $H_{I_i} = \sum_{k} A_{i;k} \otimes B_{k}^{E_i}$, i.e, as a sum of operators that act separately on the system and corresponding environment. This decomposition, together with the ansatz \eref{eq:weakcopling_appr} allows us to rewrite the overall QME \eref{eq:QME_multiple_envs} in terms of correlation functions involving only pairs of baths. Furthermore, the tensor-product structure of \eqnref{eq:weakcopling_appr} ensures that each of these reduces to a product of single-bath correlation functions. Hence, as any single-bath (one-time) correlation function can always be assumed to be zero, every summand in the cross-term of the QME \eref{eq:QME_multiple_envs} must independently vanish.

Note that, in particular, this holds for all QMEs derived using the \emph{time-convolutionless} approach \cite{Breuer2001} up to second order in all the interaction parameters representing coupling strengths for each environment.

\subsubsection{Commutativity of microscopic Hamiltonians}
\begin{figure}[!t]
{\centering \includegraphics[width=0.85\linewidth]{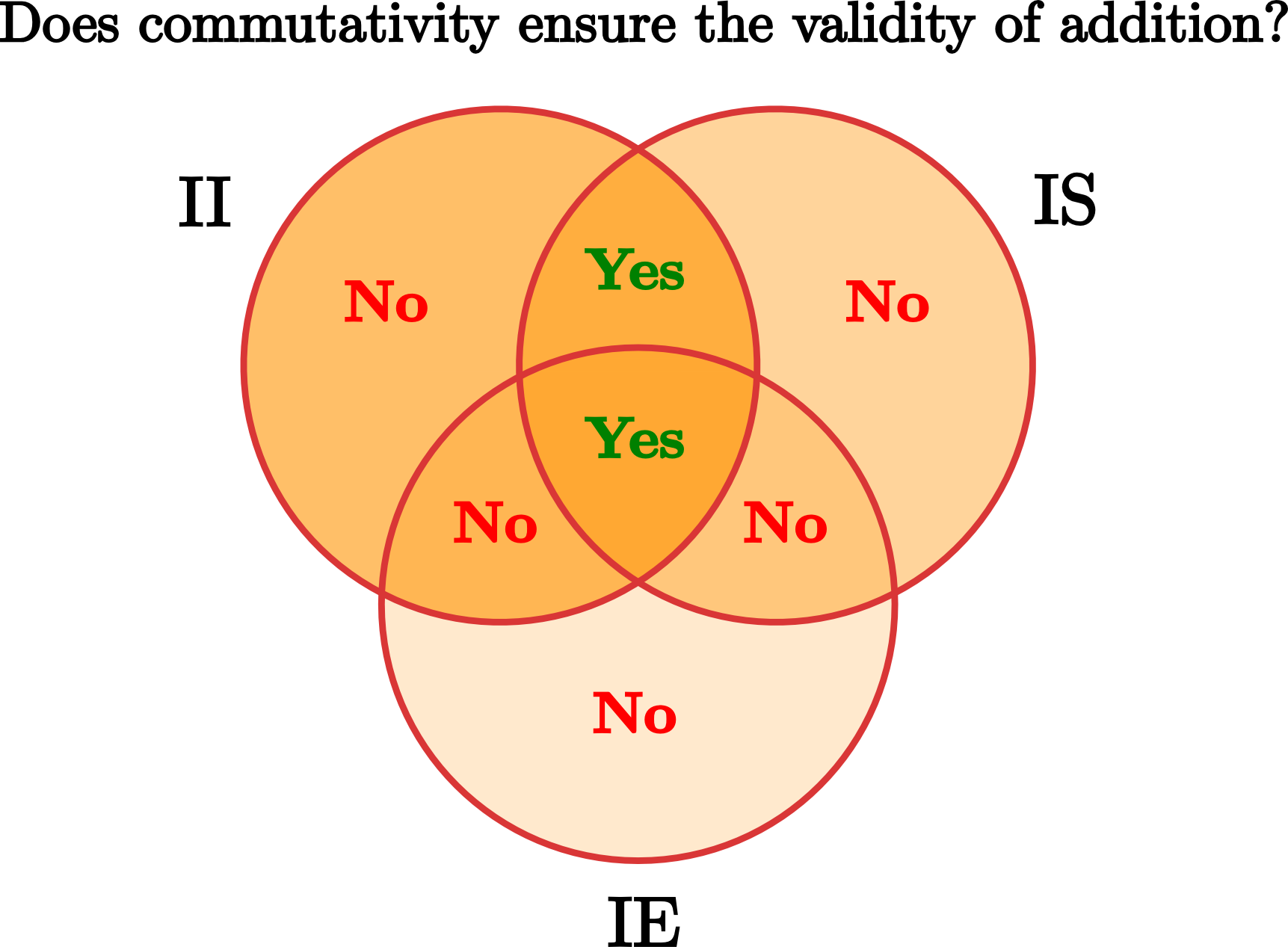}}
\caption{
\textbf{Validity of generator addition as assured by the commutativity of microscopic Hamiltonians}. The QME \eref{eq:QME_multiple_envs} describing the multiple-environment scenario of \figref{fig.settings}(b) is considered. Each set (circle) above indicates that commutativity of the interaction Hamiltonians with:~II -- each other, IS -- the system Hamiltonian, IE -- all the free Hamiltonians of environments; can be assumed.}
\label{fig:venn}
\end{figure}
Next, we investigate the implications that commutativity of the system, environment, and interaction Hamiltonians has on the validity of generator addition. We consider the cases when all $H_{I_i}$ commute with each other (II), with $H_S$ (IS), or with all the $H_{E_i}$ (IE), and summarise the results in \figref{fig:venn}. We find that:
\begin{lem}
Only when the interaction Hamiltonians commute among themselves and with the system free Hamiltonian, i.e., $[H_{I_i},H_{I_j}]=0$ and $[H_{I_i},H_S]=0$ for all $i,j$;~can the overall QME be constructed by adding dynamical generators associated individually with each environment---ignoring the cross-term in \eqnref{eq:QME_multiple_envs}. 
\label{lem:commutativity}
\end{lem}
\noindent %
Again, we summarise here the proof of \lemref{lem:commutativity} that can be found in \appref{app:valid_add_gens_comm}. However, in contrast to the discussion of the weak-coupling regime, we are required to return to the microscopic derivation of the QME \eref{eq:QME_multiple_envs}. 

Crucially, the commutativity of interaction Hamiltonians with one another as well as with $H_S$---the region $\text{II}\cap\text{IS}$ marked `Yes' in \figref{fig:venn}---assures that the unitary of the global von Neumann equation \eref{eq:ODE_S+E} factorises, i.e.:
\BE
U_{SE}(t)\!=\! \ee^{-\ii (H_S + \sum_i \!H_{E_i} \!+\! H_{I_i}) t} \!=\! \ee^{-\ii H_S t}\prod_i \ee^{-\ii (H_{I_i} + H_{E_i})t}.
\EE
As a result, the system dynamics is described by a product of commuting CPTP maps, $\bar{\rho}_S(t) = \prod_i \LambdaE_t^{(i)} [\bar{\rho}_S(0)]$, associated with each individual environment and given by $\LambdaE_t^{(i)}[\bullet] = \tr_{E_i}\!\{\ee^{-\ii(H_{I_i} + H_{E_i})t}\, (\bullet \ot\rho_{E_i})\, \ee^{\ii(H_{I_i} + H_{E_i})t}\}$. 
By differentiating the dynamics with respect to $t$, it is then evident that the QME takes the form \eref{eq:QME_multiple_envs} with each $\Lenv[(i)]_t=\dot{\LambdaE}_{t}^{(i)}\circ(\LambdaE_{t}^{(i)})^{-1}$ and the cross-term being, indeed, absent. Note that, as all $\Lenv[(i)]_t$ must then represent generator families belonging to a common commutative class, if each of them yields dynamics that is also SS, they all must belong to the same SSC class in \figref{fig:generatorgeometry}.

In all other cases marked `No' in \figref{fig:venn}, the commutativity does not ensure the generators to simply add. We demonstrate this by providing explicit counterexamples based on a concrete microscopic model, for which the evolution of the system interacting with each environment separately, as well as all simultaneously, can be explicitly solved. It is sufficient to do so for the settings in which either all $H_{I_i}$ commute with all $H_{E_i}$ and $H_S$ (intersection $\text{IS}\cap\text{IE}$ in \figref{fig:venn}), or all $H_{I_i}$ commute with each other and all $H_{E_i}$ (intersection $\text{II}\cap\text{IE}$), since it then follows that neither II, IS, nor IE alone can ensure the validity of generator addition. 

Note that it is known that $[H_I,H_S]=0$ implies the evolution to be CP-divisible \cite{Khalil2014}. Thus, as families of CP-divisible generators are additive (c.f.~\figref{fig:generatorgeometry}), our counterexample for $\text{IS}\cap\text{IE}$ below corresponds to a case where generator addition results in dynamics which is physical but does \emph{not} agree with the microscopic derivation. For a setting in which $H_{I_i}$ do not commute neither among each other nor with $H_S$, validity of adding generators has been discussed in \refcite{Chan2014}.

\section{Spin-magnet model}
\label{sec:magnets}
In this section, we construct counterexamples to the validity of adding generators at the QME level for the relevant cases summarised in \figref{fig:venn}. Inspired by \refcite{Allahverdyan2013,Perarnau2016}, we consider a single \emph{qubit} (spin-$1/2$ particle) in contact with multiple \emph{magnets}---environments consisting of many spin-$1/2$ systems. Within this model, the closed dynamics of the global qubit-magnets system can be solved and, after tracing out the magnets' degrees of freedom, the exact open dynamics of the qubit can be obtained. As a result, we can determine the QME generators describing dynamics of the qubit when coupled to one or more magnets, so that comparison with evolutions obtained by adding the corresponding generators can be explicitly made.

\subsection{Magnet as an environment}
Within our model, we allow the system free Hamiltonian $H_S$ to be chosen arbitrarily, yet, for simplicity, we take the free Hamiltonian of the environment to vanish, $H_E = 0$. As a result, any initial environment state is stationary, with $[\rho_E,H_E]=0$ trivially for any  $\rho_{E}$.

The environment is represented by a magnet that consists of $N$ spin-1/2 particles, for which we introduce the magnetisation operator:
\BE
\op m = \sum_{n=1}^{N} \sz^{(n)} = \sum_{k=0}^N m_k \, \Pi_k ,
\EE
where $\Pi_k$ is the projector onto the subspaces with magnetisation $m_k$ (i.e., with $k$ spins pointing up). The magnetisation $m_k$ takes $N+1$ equally spaced values between $-N$ and $N$:
\BE
m_k = -N +2k \qquad\trm{for}\qquad k=0,...,N.
\EE

Consistently with \secref{sec:micro}, the initial state of the spin-magnet system reads $\rhoSE(0) = \rhoS(0) \otimes \rhoE$, where we take the initial magnet state to be a classical mixture of different magnetisations, i.e.,
\BE
\rhoE = \sum_{k=0}^{N} q_k \Pi_k.
\label{eq:initial_magnet_state}
\EE
The initial probability for an observation of the magnetisation to yield $m_k$ is then
\BE
p(m_k) = \tr[\rhoE \Pi_k] = q_k \tr \Pi_k.
\label{eq:p_(m_k)}
\EE
In the limit of large $N$, $p(m_k)$ approaches a continuous distribution $p(m)$, whose moments can then be computed as follows:
\BE
\sum_{k=0}^N (m_k)^s p(m_k) 
\quad\underset{N\to\infty}{\longrightarrow}\quad
\int_{-\infty}^{\infty} dm \;  m^s \; p(m) .
\EE

We consider only interaction Hamiltonians which couple the system qubit to the magnet's magnetisation, i.e,
\BE
H_I = A \otimes\op m
\label{eq:HI_magnet}
\EE
with $A$ being an arbitrary qubit observable. However, in order to be consistent with \secref{sec:micro}, we must impose that 
\BE
\tr_E\{H_I\rho_E\}=A \tr\{\op{m}\rho_E\}=A\,\sum_{k=0}^N m_k\, p (m_k)=0,
\EE 
which implies that we must restrict to distributions $p(m_k)$ (and $p(m)$ in the $N\to\infty$ limit) with zero mean. 

In the examples discussed below, we consider two initial magnetisation distributions for the magnet in the asymptotic $N$ limit. In particular, we consider a \emph{Gaussian} distribution:
\begin{align}
p(m)= \frac{1}{\sqrt{2 \pi}\sigma} \ee^{-\frac{m^2}{2\sigma^2}},
\label{eq.gaussianmagnet}
\end{align}
which formally corresponds to the asymptotic limit of a magnet being described by a microcanonical ensemble \cite{Allahverdyan2013}---its every spin configuration being equally probable, with $q_k = 1/2^{N}$ in \eqnref{eq:p_(m_k)}, yielding a binomial distribution of magnetisation with variance equal to the number of spins ($\sigma^2=N$).
We also consider the case when the magnetisation follows a \emph{Lorentzian} distribution in the $N\to\infty$ limit, i.e.:
\begin{align}
\label{eq.lorentzianmagnet}
p(m)= \frac{\lambda}{\pi (\lambda^2+m^2)} ,
\end{align}
parametrised by the scale parameter $\lambda$ (specifying the half width at half maximum).

Given the above initial magnet state \eref{eq:initial_magnet_state} and the interaction Hamiltonian \eref{eq:HI_magnet}, the global system-magnet state constitutes at all times a mixture of states with different magnet magnetisation. In particular, it can be decomposed at any $t\ge0$ as
\BE
\label{eq.globalstate}
\rhoSE(t) = \sum_k q_k \rho_S^{(k)}(t) \otimes \Pi_k,
\EE
where every $\rho_S^{(k)}$ can be understood as the (normalised) state of the system conditioned on the magnet possessing the magnetisation $m_k$. Consequently, the full reduced system state at time $t$ reads
\BE
\rho_S(t) = \tr_E \rhoSE(t) = \sum_k p(m_k) \rho_S^{(k)}(t).
\label{eq.systemstate}
\EE

Crucially, within the model each of the conditional states $\rho_S^{(k)}$  in \eqnref{eq.systemstate} evolves independently. In order to show this, we substitute the system-environment state \eref{eq.globalstate} and the microscopic Hamiltonians into the global von Neumann equation \eref{eq:ODE_S+E} to obtain
\begin{align}
 \dot{\rho}_{SE}(t)  &=   -\ii [H_S + H_I , \rhoSE(t)] \nonumber\\
& = -\ii [H_S\otimes\1 + A\otimes\op m , \rhoSE(t)] \nonumber\\
& = -\ii [H_S\otimes\1 + \sum_k m_k A\otimes\Pi_k, \sum_l q_l \rho_S^{(l)}(t) \otimes \Pi_l ] \nonumber\\
& = -\ii \sum_k q_k [H_S + m_k A , \rho_S^{(k)}(t) ] \otimes\Pi_k.
\end{align}

As no coupling between different magnetisation subspaces (labelled by $k$) is present, after rewriting the l.h.s.~above using \eqnref{eq.globalstate}, one obtains a set of uncoupled differential equations for each conditional state:
\BE
\dot{\rho}_S^{(k)}(t) = -\ii [H_S + m_k A , \rho_S^{(k)}(t) ],
\label{eq:vN_conditional}
\EE
with $\rho_S^{(k)}(t)=\rho_S(0)$ for each $k$. 

Hence, every $\rho_S^{(k)}$ evolves unitarily within our model with $U_S^{(k)}(t):=\exp[-\ii (H_S+m_kA)t]$, while the overall evolution of the qubit \eref{eq.systemstate} is given by the dynamical map, $\Lambda_t$, corresponding to a mixture of such (conditional) unitary transformations distributed according to the initial magnetisation distribution of the magnet, $p(m_k)$, i.e.:
\BE
\rho_S(t) = \Lambda_t[\rho_S(0)]=\sum_k p(m_k)\;U_S^{(k)}(t)\,\rho_S(0)\, U_S^{(k)\dagger}(t).
\label{eq:spin_reduced_state}
\EE
Furthermore, as $\Lambda_t$ constitutes a mixture of unitaries in the model, it must be unital, i.e., for all $t\ge0$:~$\Lambda_t[\1]=\1$.

\subsubsection{Bloch ball representation}
We rewrite the above qubit dynamics employing the Bloch ball representation \cite{Bengtsson2006}, i.e., $\rho \equiv \frac{1}{2}(\1 + \rV\cdot\sV)$ with the Bloch vector $\rV$ unambiguously specifying a qubit state $\rho$. Then, \eqnsref{eq:vN_conditional}{eq:spin_reduced_state} read, respectively:
\BE
\label{eq.commeqn}
\dot{\rV}^{(k)}(t)\cdot\sV = -\ii [H_S + m_k A \, , \, \rV^{(k)}(t)\cdot\sV ]
\EE
and
\BE
\rV(t) = \mat{D}_t\,\rV(0) = \left[\sum_k p(m_k)\, \mat{R}^{(k)}(t)\right] \rV(0).
\label{eq:bloch_vector}
\EE
The rotation matrices above, $\mat{R}^{(k)}$, constitute the SO(3) representations of the unitaries $U_S^{(k)}\in \text{SU(2)}$ in \eqnref{eq:spin_reduced_state}, and are thus similarly mixed according to $p(m_k)$. 

The qubit dynamical map, $\Lambda_t$ of \eqnref{eq:spin_reduced_state}, is represented by an affine transformation of the Bloch vector:
\BE
\mat{D}_t := \sum_k p(m_k)\, \mat{R}^{(k)}(t)
\;\underset{N\to\infty}{\longrightarrow}\;
\int_{-\infty}^{\infty}\!\!\! dm \; p(m)\, \mat{R}(m,t),
\label{eq:affine_map}
\EE
which is linear due to $\Lambda_t$ being unital within the magnet model, i.e., does not contain a translation.

Now, as the spaces of physical $\Lambda_t$ (dynamical maps) and $\mat{D}_t$  (affine transformations) are isomorphic \cite{Bengtsson2006}, the dynamical generators of the former $\LSP_t:=\dot{\Lambda}_t \circ\Lambda_t^{-1}$ (see \appref{app:dyn_gen_descr}) directly translate onto $\LSPb_t:=\dot{\mat{D}}_t \,\mat{D}_t^{-1}$ of the latter, with the map composition and inversion replaced by matrix multiplication and inversion, respectively. Moreover, as the vector spaces containing families of generators defined in this manner must also be isomorphic, all the notions described in \secref{sub:gen_add}---in particular, rescalability and additivity---naturally carry over.

However, in order to define the Bloch ball representation of the environment-induced generator $\Lenv_t$ in the QME \eref{eq:QME}, we must correctly relate it to the interaction and Schr\"{o}dinger pictures of the dynamics, summarised in \appref{app:SP_and_IP}.  In general (see \appref{app:QME_SP} for the derivation from the dynamical maps perspective), the Bloch ball representation of $\Lenv_t$ reads
\BE
\Lenvb_t := \LSPb_t - \dot{\mat{R}}_S(t)\,\mat{R}_S(t)^{-1}
=
\mat{R}_S(t)\,  \LIPb_t \,\mat{R}_S(t)^{-1},
\label{eq:Lb_env_ind}
\EE
where $\mat{R}_S(t)\in\text{SO(3)}$ is the rotation matrix of the Bloch vector that represents the qubit unitary map, $U_S(t):=\exp[-\ii H_S t]\in\text{SU(2)}$, induced by the system free Hamiltonian $H_S$. 

As stated in \eqnref{eq:Lb_env_ind}, $\Lenvb_t$  may be equivalently specified with help of $\LIPb_t:=\dot{\bar{\mat{D}}}_t \,\bar{\mat{D}}^{-1}_t$, i.e, the Bloch ball representation of the dynamical generator defined in the interaction picture, $\LIP_t:=\dot{\bar{\Lambda}}_{t}\circ\bar{\Lambda}_{t}^{-1}$---see also \appref{app:QME_TL} for its formal microscopic definition. Importantly, $\LIPb_t$ may be directly computed for a given $\mat{D}_t$ of \eqnref{eq:affine_map} by first transforming it to the interaction picture, i.e., determining $\bar{\mat{D}}_t:=\mat{R}_S^{-1}(t)\,\mat{D}_t$ that is the Bloch-equivalent of $\bar{\Lambda}_t[\bullet]:=U_S^\dagger(t)\,\Lambda_{t}\!\left[\bullet\right] U_S(t)$ discussed in \appref{app:SP_and_IP}. 

On the other hand, as $\Lenvb_t$ and $\LIPb_t$ are linearly related via \eqnref{eq:Lb_env_ind}, their vector spaces must be isomorphic. Hence, in what follows, we may  equivalently stick to the interaction picture and consider $\LIPb_t$ instead, in particular, $\LIPb[(1)]_t+\LIPb[(2)]_t\;\Leftrightarrow\;\Lenvb[(1)]_t+\Lenvb[(2)]_t$, when verifying the validity of generator addition.

\subsubsection{Example: magnet-induced dephasing}
\label{sec:magnet_dephasing}
To illustrate the model, we first consider the case when it is employed to provide a simple microscopic derivation of the qubit \emph{dephasing dynamics}. We take:
\BE
\label{eq.magmodelonebath}
H_{E}=0 ,\quad H_S=0 , \quad H_{I} = \frac{1}{2} g \,\sz \otimes \op m,
\EE
with the system Hamiltonian being absent, so that all the generators, $\LSPb_t=\LIPb_t=\Lenvb_t$ in \eqnref{eq:Lb_env_ind}, become equivalent. From \eqnref{eq.commeqn} we get
\BE
\dot{\rV}^{(k)}(t)\cdot\sV = -\frac{\ii}{2} g m_{k} [ \sz \, , \, \rV^{(k)}(t)\cdot\sV ] ,
\EE
which just yields rotations of the Bloch ball around the $z$ axis with angular speed depending on the magnetisation $m_k$, i.e., \eqnref{eq:affine_map} with (in Cartesian coordinates)
\BE
\label{eq.singlebathsol}
\mat{R}(m,t) = 
\begin{pmatrix}
 \cos \left(g m t\right) & -\sin \left(g m t\right) & 0 \\
 \sin \left(g m t\right) & \cos \left(g m t\right) & 0 \\
 0 & 0 & 1 \\
\end{pmatrix}
\EE
for the $N\to\infty$ limit.

Integrating \eqnref{eq.singlebathsol} over the initial magnetisation distribution $p(m)$, we obtain the affine transformation \eref{eq:affine_map} to asymptotically read
\BE
\mat{D}_t =
\begin{pmatrix}
 \ee^{- f(t)} & 0 & 0 \\
 0 & \ee^{- f(t)}  & 0 \\
 0 & 0 & 1 \\
\end{pmatrix},
\EE
with $f(t)=\frac{1}{2} \sigma^2 g^2 t^2$ and $f(t)=\gamma g t$ in case of the Gaussian \eref{eq.gaussianmagnet} and Lorentzian \eref{eq.lorentzianmagnet} distributions $p(m)$, respectively. Hence, the corresponding generators \eref{eq:Lb_env_ind} in Bloch ball representation take a simple form:
\BE
\LSPb_t =
\begin{pmatrix}
 - \gamma(t) & 0 & 0 \\
 0 &  - \gamma(t) & 0 \\
 0 & 0 & 0 \\
\end{pmatrix},
\EE
which corresponds to the standard \emph{dephasing generator} with a time-dependent rate (as defined in \eqnref{eq:dyn_gens_dephasing} of \appref{app:dyn_gens_qubit_dyns}), i.e.:
\BE
\LSP_t[\bullet] = \gamma(t)\,(\sz \bullet \sz - \bullet)
\EE
with $\gamma(t)=\sigma^2 g^2 t$ in the Gaussian case, and constant (semigroup) $\gamma(t)=\lambda g$ in the Lorentzian case.

\subsection{Counterexamples to sufficiency of the commutativity assumptions}
\label{subsec:Multiple}
We now prove the regions marked `No' in \figref{fig:venn}. In particular, we provide explicit counterexamples which assure that the commutativity assumption---associated with the particular region of the Venn diagram---is generally \emph{not} sufficient for the system dynamics to be recoverable by simple addition of the generators attributed to each of the environments. In order to do so, it is enough to consider the scenario in which the qubit is independently coupled to just two magnets via the mechanism described above.

\subsubsection{$\text{IS}\cap\text{IE}$ commutativity assumption}
\label{app:counter_ISandIE}
We start with an example of dynamics, in which the interaction Hamiltonians (trivially) commute with the free system and all the environmental Hamiltonians, but not among each other. In particular, we simply set:
\BEA
H_S & =& H_{E_1}=H_{E_2}=0, \label{eq.magmodelonlyint} \\
H_{I_1} &=& \frac{1}{2} g_1 \,\sz \otimes \op m_1 , 
\quad
H_{I_2} = \frac{1}{2} g_2\, \sx \otimes \op m_2 \nonumber
\EEA
with subscripts $\{1,2\}$ labelling to the first and the second magnet. As in the case of the dephasing-noise derivation above, the system Hamiltonian is absent, so the generators in \eqnref{eq:Lb_env_ind} coincide with $\LSPb_t=\LIPb_t=\Lenvb_t$.

In the case of simultaneous coupling to two magnets, \eqnref{eq.globalstate} naturally generalises to
\BE
\rho_{SE_1E_2}(t) = \sum_{k,k'} q_{k,k'} \,\rho_S^{(k,k')}(t) \otimes \Pi_k \otimes \Pi_{k'} ,
\EE
where $q_{k,k'}$ now represents the joint probability of finding the first and the second magnet in magnetisations $m_k$ and $m_{k'}$, respectively, 
while $\rho_S^{(k,k')}(t)$ stands for the corresponding conditional reduced state of the system.

Consequently, the (conditional) von Neumann equation \eref{eq.commeqn}, which now must be derived for $H_I=H_{I_1}+H_{I_2}$, describes the dynamics of Bloch-vectors that represent each conditional state, $\rho_S^{(k,k')}(t)$, being also parametrised by the two indices $k$ and $k'$, i.e.:
\BE
\dot{\rV}^{(k,k')}(t)\cdot\sV = -\frac{\ii}{2}[g_1 m_{1,k} \sz + g_2 m_{2,k'} \sx \, , \, \rV^{(k,k')}(t)\cdot\sV ].
\label{eq:Bloch_dyn_case1}
\EE
\eqnref{eq:Bloch_dyn_case1} leads to coupled equations in the Cartesian basis, i.e. (dropping the indices $k,k'$ and the explicit time-dependence for simplicity):
\begin{subequations}
\label{firsteqsmotion}
\begin{align}
\dot{r}_x & = - g_1 m_1 r_y , \\
\dot{r}_y & = g_1 m_1 r_x - g_2 m_2 r_z , \\
\dot{r}_z & = g_2 m_2 r_y,
\end{align}
\end{subequations}
which can be analytically solved to obtain the $\mat{R}$-matrix in \eqnref{eq:affine_map}---labelled $\mat{R}_{12}$ to indicate that both magnets are involved. $\mat{R}_{12}(m_1,m_2,t)$ possesses now two magnetisation parameters associated with each of the magnets, and we state its explicit form in \appref{app:IS_IE}.

Furthermore, we can straightforwardly obtain the solution of the equations of motion when only one of the magnets is present by simply setting either $g_1=0$ or $g_2=0$ in $\mat{R}_{12}$. In presence of only the first magnet ($g_2=0$), we recover the magnet-induced dephasing noise described above---with $\mat{R}_{12}(m_1,m_2,t)$ simplifying to $\mat{R}_1(m_1,t)$ that takes exactly the form \eref{eq.singlebathsol}. On the other hand, when only the second magnet ($g_1=0$) is present, which couples to the system via $\sx$ rather than $\sz$, see \eqnref{eq.magmodelonlyint}, we obtain $\mat{R}_2(m_2,t)$ as in \eqnref{eq.singlebathsol} but with coordinates cyclically exchanged---see \appref{app:IS_IE} for explicit expressions.

We then average each $\mat{R}_\msf{x}$, where $\msf{x}=\{12,1,2\}$ denotes the magnet(s) being present, over the initial magnetisations in the $N\!\to\!\infty$ limit. This way, we obtain the affine maps \eref{eq:affine_map} representing the corresponding qubit dynamics for all the three cases in the asymptotic $N$ limit as:
\BE
\mat{D}^{(\msf{x})}_t = \int_{-\infty}^{\infty} \!\!\!\dd m_\msf{x} \;
p(m_\msf{x}) \; \mat{R}_\msf{x}(m_\msf{x},t),
\label{eq.numint}
\EE
where in presence of both magnets $m_{12}\equiv(m_1,m_2)$ and $p(m_{12})\equiv p(m_1)\,p(m_2)$;~and we take each $p(m_i)$ to follow a Gaussian distribution \eref{eq.gaussianmagnet} with variance $\sigma_i$. Similarly, we obtain the integral expressions for the time-derivatives of the affine maps, $\dot{\mat{D}}_t^{(\msf{x})}=\int \!\dd m_\msf{x}\, p(m_\msf{x})\, \dot{\mat{R}}_\msf{x}(m_\msf{x},t)$ after also computing analytically all the corresponding $\dot{\mat{R}}_{\msf{x}}$.

Finally, we choose particular values of $g_1$, $g_2$, $\sigma_1$, $\sigma_2$ and time $t$, in order to numerically compute the integrals over the magnetisation parameters and obtain all $\mat{D}_t^{(\msf{x})}$ and $\dot{\mat{D}}_t^{(\msf{x})}$. Our choice, allows us then to explicitly construct dynamical generators $\LSPb[(\msf{x})]_t = \dot{\mat{D}}_t^{(\msf{x})}(\mat{D}_t^{(\msf{x})})^{-1}$, which importantly exhibit $\LSPb[(12)]_t \neq \LSPb[(1)]_t + \LSPb[(2)]_t$, see \appref{app:IS_IE}.

Hence, we conclude that the commutativity of the interaction Hamiltonians with both system and environment Hamiltonians, but not with each other, \emph{cannot} assure the generators to simply add at the level of the QME---as denoted by the `No' label in the region of \figref{fig:venn} representing the $\text{IS}\cap\text{IE}$ commutativity assumption.

\subsubsection{$\text{II}\cap\text{IE}$ commutativity assumption}
\label{app:counter_IIandIE}
In order to construct an example of reduced dynamics in which, at the microscopic level, the interaction Hamiltonians commute with each other and all free environmental Hamiltonians but not the system Hamiltonian, we consider again a two-magnet model but this time set:
\BEA
H_S &=& \frac{1}{2}\omega \sx, \quad H_{E_1}=H_{E_2}=0, \label{eq.magmodel2} \\
H_{I_1} &=& \frac{1}{2} g_1 \,\sz \otimes \op m_1 , 
\quad
H_{I_2} = \frac{1}{2} g_2\, \sz \otimes \op m_2 . \nonumber
\EEA
In contrast to the previous example, since $H_S\neq0$, in order to investigate the validity of generator addition we must either consider the environment-induced generators $\Lenvb_t$ defined in \eqnref{eq:Lb_env_ind} or the generators $\LIPb_t$ directly computed in the interaction picture.

We do the latter and compute both $H_{I_i}$ in the interaction picture, i.e., $\bar{H}_{I_i}(t)$, which are obtained by replacing $\sz$ Pauli operators in \eqnref{eq.magmodel2} with
\BE
\bsz(t):=\ee^{\ii H_S t} \sz \ee^{-\ii H_S t} = \cos(\omega t) \sz  + \sin(\omega t) \sy.
\EE
Importantly, $\bsz(t)$ should be interpreted as a (time-dependent) operator $A$ in the general expression \eref{eq:HI_magnet} for $H_I$ that, in contrast to the previous case, is now identical for both magnets. Thus, inspecting the general expression for the dynamics \eref{eq.commeqn}, we obtain the equation of motion for the Bloch vector in the interaction picture, $\bar{\rV}^{(k,k')}(t):=\mat{R}_S^{-1}(t)\,\rV^{(k,k')}(t)$, that represents the qubit state conditioned on the first and second magnet possessing magnetisations $m_k$ and $m_{k'}$, respectively, as
\begin{align}
\dot{\bar{\rV}}^{(k,k')}(t)\cdot\sV =& -\frac{\ii}{2} (g_1 m_{1,k} + g_2 m_{2,k'})\times
\label{eq:Bloch_dyn_case2}\\
&\quad\times [\cos(\omega t) \sz  + \sin(\omega t) \sy \, , \, \bar{\rV}^{(k,k')}(t)\cdot\sV ] , 
\nonumber
\end{align}
which leads to coupled equations (again, dropping the indices $k,k'$ and the explicit time-dependence):
\begin{subequations}
\label{secondeqsmotion}
\begin{align}
\dot{\bar{r}}_x & = (g_1 m_1 + g_2 m_2) (\sin(\omega t) \bar{r}_z - \cos(\omega t) \bar{r}_y) , \\
\dot{\bar{r}}_y & = (g_1 m_1 + g_2 m_2) \cos(\omega t) \bar{r}_x  , \\
\dot{\bar{r}}_z & = - (g_1 m_1 + g_2 m_2) \sin(\omega t) \bar{r}_x.
\end{align} 
\end{subequations}

As before, see \appref{app:II_IE}, we solve the above equations of motion in order to obtain the $\bar{\mat{R}}$-matrix of \eqnref{eq:affine_map} in the interaction picture, i.e., $\bar{\mat{R}}_{12}(m_1,m_2,t)$. Again, by setting either $g_2=0$ or $g_1=0$, we obtain expressions for $\bar{\mat{R}}_1$ and $\bar{\mat{R}}_2$, respectively, corresponding to the cases when only first or second magnet is present. We then also compute all $\dot{\bar{\mat{R}}}_\msf{x}$ with $\msf{x}=\{12,1,2\}$, in order to arrive at integral expressions for both the affine maps and their time-derivatives, i.e., $\bar{\mat{D}}_t^{(\msf{x})}$ and $\dot{\bar{\mat{D}}}_t^{(\msf{x})}$, respectively, computed now in the interaction picture.

As in the previous example, we take initial magnetisation distributions of both magnets to be Gaussian and fix all the model parameters (i.e., $\omega$ and $g_1$, $g_2$  for the system and interaction, $\sigma_1$, $\sigma_2$ for the magnets, as well as the time $t$) in order to numerically perform the integration over magnetisations $m_\msf{x}$. We then find, see \appref{app:II_IE}, a choice of parameters for which it is clear that the dynamical generators, $\LIPb[(\msf{x})]_t = \dot{\bar{\mat{D}}}_t^{(\msf{x})}(\bar{\mat{D}}_t^{(\msf{x})})^{-1}$, fulfil $\LIPb[(12)]_t \neq \LIPb[(1)]_t + \LIPb[(2)]_t$.

Hence, we similarly conclude that the commutativity of all the interaction Hamiltonians with each other, and all the free Hamiltonians of environments also \emph{cannot} assure the generators to simply add at the level of the QME---proving the `No' label in the region of \figref{fig:venn} representing the $\text{II}\cap\text{IE}$ commutativity assumption.

\section{Conclusions}
\label{sec:conclusion}
We have investigated under what circumstances modifications to open system dynamics can be effectively dealt with at the master equation level by adding dynamical generators. We have identified a condition---semigroup simulability and commutativity preservation---applicable beyond Markovian (CP-divisible) dynamics which guarantees generator addition to yield physical evolutions. We have also demonstrated by considering simple qubit generators that even mild violation of this condition may yield unphysical dynamics under generator addition.

Moreover, even when physically valid, generator addition does not generally correspond to the real evolution derived from a microscopic model describing interactions with multiple environments. We have formulated a general criterion under which the addition of generators associated with each individual environment yields the correct dynamics. 

We have then shown that this condition is generally satisfied in the weak-coupling regime, whenever it is correct to use a master equation derived assuming a tensor-product ansatz for the global state describing the system and environments. Finally, we have demonstrated that, at the microscopic level, the commutativity of interaction Hamiltonians among each other and with the system Hamiltonian also ensures addition of dynamical generators to give the correct dynamics. 

We believe that our results may prove useful in areas where the master equation description of open quantum systems is a common workhorse, including quantum metrology, thermodynamics, transport, and engineered dissipation.

\begin{acknowledgments}
We would like to thank L.~Aolita and N.~Bernades for interesting discussions on non-Markovianity that sparked this work, as well as A.~Smirne, M.~Lostaglio, S.~Huelga, L.~Correa and A.~S.~S{\o}rensen for helpful exchanges. J.K.~and B.B.~acknowledge support from the Spanish MINECO (Grant QIBEQI FIS2016-80773-P and Severo Ochoa SEV-2015-0522), Fundaci\'{o} Privada Cellex, Generalitat de Catalunya (SGR875 and CERCA Program). J.K.~is also supported by the EU Horizon 2020 programme under the MSCA Fellowship Q-METAPP (no. 655161), while B.B.~by the ICFO-MPQ Fellowship. M.P.-L.~acknowledges also support from the Alexander von Humboldt Foundation.
\end{acknowledgments}

\bibliography{sum_gens}

\appendix

\section{QMEs as families of dynamical generators}
\label{app:QMEs_dyn_gen_fams}

\subsection{Describing open system dynamics}
\label{app:dyn_gen_descr}

\subsubsection{Physically valid quantum dynamics}
A particular evolution of an open quantum system is formally represented by a continuous family of density matrices $\left\{ \rho_{t}\in\cB(\mathcal{H}_{d}) \right\} _{t\ge0}$ that describe the system state at each time $t\ge0$. The system evolution is then defined by a \emph{family of dynamical maps} (quantum channels \cite{Bengtsson2006}), $\left\{ \Lambda_{t}\right\} _{t\ge0}$ with $\Lambda_{0}=\cI$ being the identity map, such that for any initial system state, $\rho_{0}$, the state at time $t\ge0$ is given by
\BE
\rho_{t}=\Lambda_{t}\!\left[\rho_{0}\right].
\label{eq:dynamics}
\EE
Importantly, for a given dynamics to be \emph{physical} the family $\left\{ \Lambda_{t}\right\} _{t\ge0}$ must consist of \emph{completely-positive trace preserving} (CPTP) maps. Only then, for any given enlarged initial state $\varrho(0)\in\cB(\cH_d\ot\cH_d')$ with arbitrary $d'=\dim\cH_d'$, the state at every time $t\ge0$, i.e., $\varrho(t)=\Lambda_t\ot\cI[\varrho(0)]$, is guaranteed to be correctly described by a positive semidefinite matrix.

In practice, any linear map $\Lambda:\,\mathcal{B}(\mathcal{H}_{d})\to\mathcal{B}(\mathcal{H}_{d'})$ may be verified to be CPTP by constructing its corresponding \emph{Choi-Jamio\l kowski} (CJ) \emph{matrix} $\Omega_{\Lambda}\in\mathcal{B}(\mathcal{H}_{d'}\otimes\mathcal{H}_{d})$ defined as \citep{Jamiolkowski1972,Choi1975}:
\BE
\Omega_{\Lambda}:=\Lambda\ot\cI\left[\left|\psi\right\rangle \!\left\langle \psi\right|\right],
\label{eq:CJ_matrix}
\EE
with $\left|\psi\right\rangle =\sum_{i=1}^{d}\left|i\right\rangle \!\left|i\right\rangle$ and $\left\{ \left|i\right\rangle \right\} _{i=1}^{d}$ being some orthonormal basis spanning $\mathcal{H}_{d}$. In particular, a map $\Lambda$ is CP and TP iff its CJ matrix is positive semi-definite, i.e., $\Omega_{\Lambda}\ge0$, and satisfies $\mathrm{Tr}_{\mathcal{H}_{d'}}\!\left\{ \Omega_{\Lambda}\right\} =\openone_{d}$, respectively.

\subsubsection{Dynamical generators}
One may associate with any given dynamics \eref{eq:dynamics} the \emph{family of dynamical generators}, $\left\{ \cL_t\right\} _{t\ge0}$, that specify for each $t\ge0$ the time-local QME stated in \eqnref{eq:QME_dyn_gen} of the main text \citep{Andersson2007,Chruscinski2014,Hall2014}, i.e.:
\BE
\dot{\rho}_{t}=\cL_t\!\left[\rho_{t}\right]
\quad\Big(\Longleftrightarrow\;
\dot{\Lambda}_{t}=\cL_t\circ\Lambda_{t}\Big)
\label{eq:master_equation}
\EE
with $\dot{\bullet}\equiv\frac{{\dd}}{{\dd}t}\bullet$ and the \emph{dynamical generator} being then formally defined at each $t\ge0$ as
\BE
\cL_t:=\dot{\Lambda}_{t}\circ\Lambda_{t}^{-1},
\label{eq:dyn_gen_def}
\EE
where $\Lambda_{t}^{-1}$ is the inverse ($\Lambda_{t}^{-1}\circ\Lambda_{t}=\cI$) of the dynamical map at time $t$, and is not necessarily CPTP. However, $\Lambda_{t}^{-1}$ in general may cease to exist at certain time-instances, at which the corresponding generators $\cL_{t}$ then become singular, even though the family of maps is perfectly smooth. Nevertheless, under particular conditions \citep{Andersson2007}, the resulting QME \eref{eq:master_equation} can still be integrated and yield correctly the original dynamics \eref{eq:dynamics}.

In the other direction, given a family of dynamical generators $\left\{ \cL_t\right\} _{t\ge0}$, one may write the corresponding dynamical map at any time $t$ with help of a time-ordered exponential, expressible in the Dyson-series
form \citep{Andersson2007,Chruscinski2014}:
\BE
\Lambda_{t} = \mathcal{T}_{\leftarrow}\exp\!\left\{ \int_{0}^{t}\cL_{\tau}{\dd}\tau\right\} =\sum_{i=0}^{\infty}\mathcal{S}_{t}^{(i)}\left[\cL_{\bullet}\right],
\label{eq:map_via_dyn_gens}
\EE
where $\mathcal{S}_{t}^{(0)}\left[\bullet\right]=\mathcal{I}$ and
for all $i\ge1$:
\AL{
\mathcal{S}_{t}^{(i)}\left[\cL_{\bullet}\right]
&:=
\frac{1}{i!}\int_{0}^{t}\!\!\dd t_{1}\int_{0}^{t}\!\!{\dd}t_{2}\,\dots\!\int_{0}^{t}\!\!{\dd}t_{i}\;
\mathcal{T}_{\leftarrow}\,\cL_{t_{1}}\!\circ\!\cL_{t_{2}}\!\circ\!\dots\!\circ\!\cL_{t_{i}} \nonumber\\
&=
\int_{0}^{t}\!\!{\dd}t_{1}\int_{0}^{t_{1}}\!\!{\dd}t_{2}\,\dots\!\int_{0}^{t_{i-1}}\!\!\!\!{\dd}t_{i}\;
\cL_{t_{1}}\!\circ\!\cL_{t_{2}}\!\circ\!\dots\!\circ\!\cL_{t_{i}}.
\label{eq:dupa_blada}
}

\subsubsection{Instantaneous generators}
\label{app:dyn_gens_fams_inst}
Given a family of dynamical maps $\left\{ \Lambda_{t}\right\} _{t\ge0}$ defining the evolution \eref{eq:dynamics}, one can also construct the corresponding \emph{family of instantaneous generators} $\left\{ \cX_t \right\} _{t\ge0}$, defined for each $t\ge0$ as \citep{Chruscinski2014}:
\BE
\cX_t:=\frac{\dd}{\dd t}\log\Lambda_{t}=\dot{\mathcal{Z}}_{t}
\quad\text{with}\quad
\cZ_t:=\log\Lambda_{t}
\label{eq:inst_gens_exps}
\EE
being the so-called \emph{instantaneous exponent}, such that: 
\BE
\Lambda_{t}=\exp\cZ_t=\exp\!\left\{ \int_{0}^{t}{\dd}\tau\mathcal{X}_{\tau}\right\},
\label{eq:map_via_inst_gens}
\EE
where the above expression, in contrast to \eqnref{eq:map_via_dyn_gens}, does not involve the time-ordering operator $\mathcal{T}_{\leftarrow}$.

Although the instantaneous generators $\cX_t$ cannot be directly used to construct the QME \eref{eq:master_equation}, the family of dynamical generators $\{\cL_t\}_{t\ge0}$ can be formally related to the family of instantaneous ones. In particular, by substituting into \eqnref{eq:dyn_gen_def}
\BE
\dot{\Lambda}_{t}=\frac{{\dd}}{{\dd}t}\exp\cZ_t=\int_{0}^{1}\!{\dd}s\;\;\ee^{s\cZ_t}\circ\cX_t\circ\ee^{\left(1-s\right)\cZ_t}
\EE
and $\Lambda_{t}^{-1}=\ee^{-\cZ_t}$, one observes that
\BE
\cL_t=\int_{0}^{1}\!{\dd}s\;\;\ee^{s\cZ_t}\circ\cX_t\circ\ee^{-s\cZ_t}.
\label{eq:back_trans}
\EE

\subsubsection{Commutative dynamics}
\label{app:dyn_gens_fams_comm}
A dynamics is defined to be \emph{commutative} if the maps describing the evolution in \eqnref{eq:dynamics}, or equivalently---as follows from \eqnref{eq:map_via_dyn_gens}---all the dynamical generators defining the QME \eref{eq:master_equation} commute between one another, i.e., for all $s,t\ge0$:
\BE
\left[\Lambda_{s},\Lambda_{t}\right]=0
\quad\Longleftrightarrow\quad
\left[\cL_{s},\cL_t\right]=0.
\EE

Moreover, \eqnref{eq:map_via_inst_gens} implies then that also all instantaneous exponents must commute with each another, $\left[\mathcal{Z}_{s},\cZ_t\right]=0$, but also with instantaneous generators, $0=\partial_{s}\left[\mathcal{Z}_{s},\cZ_t\right]=\left[\mathcal{X}_{s},\cZ_t\right]$. Hence, in case of commutative dynamics the dynamical and instantaneous generators must coincide at all times, as
\BE
\cL_t=\int_{0}^{1}\!{\dd}s\;\;\ee^{s\cZ_t}\circ\cX_t\circ\ee^{-s\cZ_t}=\int_{0}^{1}\!{\dd}s\;\cX_t=\cX_t,
\label{eq:inst_dyn_gens_coinc_for_comm_dyn}
\EE
so that \eqnsref{eq:map_via_dyn_gens}{eq:map_via_inst_gens} both, respectively, read
\BE
\Lambda_{t}=\exp\!\left\{ \int_{0}^{t}{\dd}\tau\cL_{\tau}\right\}=\exp\!\left\{ \int_{0}^{t}{\dd}\tau\mathcal{X}_{\tau}\right\}.
\label{eq:map_via_gens_comm}
\EE

\subsubsection{CP-divisible dynamics}
\label{app:dyn_gens_fams_CPdiv}
The dynamics is said to be divisible into CPTP maps---\emph{CP-divisible}
(or Markovian \citep{Rivas2014,Breuer2016,deVega2017})---if its corresponding family of maps,
$\left\{ \Lambda_{t}\right\} _{t\ge0}$, satisfies for all $0\le s\le t$:
\BE
\Lambda_{t}=\tilde{\Lambda}_{t,s}\circ\Lambda_{s},\label{eq:CP_div}
\EE
where $\tilde{\Lambda}_{t,s}$ is a CPTP map itself. 

At the level of dynamical generators, this is equivalent to the statement that all $\left\{ \cL_t\right\} _{t\ge0}$
are of the \emph{Gorini-Kosakowski-Sudarshan-Linblad} (GKSL) form \citep{Gorini1976,Lindblad1976}:
\BE
\cL_t\!\left[\bullet\right] =-\ii\left[H_{t},\bullet\right]+\Phi_{t}\!\left[\bullet\right]-\frac{1}{2}\left\{ \Phi_{t}^{\star}\!\left[\openone\right],\bullet\right\} ,
\label{eq:L_GKSL}
\EE
where $H_{t}$ is a time-dependent Hermitian operator, $\Phi_{t}$ is a completely-positive (CP) map, and they, respectively, represent the Hamiltonian, $\cH_t$, and dissipative, $\cD_t$, parts in the QME \eref{eq:QME_dyn_gen} of the main text. $\Phi_{t}^{\star}$ is the dual map of $\Phi_{t}$ that---given a Kraus representation of the CP map $\Phi_{t}$, i.e., a set of operators $\left\{ V_{i}(t)\right\} _{i}$ satisfying $\sum_{i}V_{i}(t)^{\dagger}V_{i}(t)=0$ for all $t\ge0$ such that $\Phi_{t}\!\left[\bullet\right]=\sum_{i}V_{i}(t)\bullet V_{i}(t)^{\dagger}$ \citep{Sudarshan1961}---is defined as $\Phi_{t}^\star\!\left[\bullet\right]:=\sum_{i}V_{i}^\dagger(t)\bullet V_{i}(t)$. 

Thus, one may rewrite \eqnref{eq:L_GKSL} also as \citep{Breuer}:
\BE
\cL_t\!\left[\bullet\right]=-\ii\left[H_{t},\bullet\right]+\sum_{i}V_{i}(t)\bullet V_{i}(t)^{\dagger}-\frac{1}{2}\left\{ V_{i}(t)^{\dagger}V_{i}(t),\bullet\right\},
\EE
which---after fixing a particular orthonormal basis of matrices $\left\{ F_{j}\right\} _{j}$ satisfying $\tr\!\left\{ F_{i}^{\dagger}F_{j}\right\} =\delta_{ij}$ in which each $V_{i}(t)=\sum_{j}\msf{V}_{ij}(t)F_{j}$---can be further rewritten as in \eqnref{eq:gen_mat_basis} of the main text:
\BE
\cL_t\!\left[\bullet\right]=-\ii\left[H_{t},\bullet\right]+\sum_{i,j}\msf{D}_{ij}(t)\left(F_{j}\bullet F_{i}^{\dagger}-\frac{1}{2}\left\{ F_{i}^{\dagger}F_{j},\bullet\right\} \right)
\label{eq:L_GKSL_D}
\EE
with the time dependence of the dissipative part being now fully contained within the matrix $\msf{D}(t)$. 

Although any dynamical generator $\cL_t$, constituting a traceless and Hermiticity-preserving operator, can be decomposed as above, the GKSL form \eref{eq:L_GKSL} ensures that for all $t\ge0$ there exists a matrix $\msf{V}(t)$ such that $\msf{D}(t)=\msf{V}(t)^{\dagger}\msf{V}(t)$. Hence, it follows that any dynamics is CP-divisible iff one may at all times decompose its corresponding dynamical generators according to \eqnref{eq:L_GKSL_D} with some positive semi-definite $\msf{D}(t)\ge0$.

\subsubsection{Semigroup dynamics}
\label{subsec:CP_div_semi}
An important subclass of commutative and CP-divisible dynamics are \emph{semigroups}, for which the whole evolution is determined by a single fixed generator $\cL$,
\BE
\left\{ \cL_t \equiv \cL \right\} _{t\ge0}\quad\implies\quad\left\{ \Lambda_{t}=\exp\!\left[t\,\cL \right] \right\} _{t\ge0},
\label{eq:semigroup}
\EE
which in order to describe physical dynamics (so that all $\Lambda_t$ are CPTP) \emph{must} be of the GKSL form \eref{eq:L_GKSL} with both the Hamiltonian $H$ and the positive semi-definite matrix $\msf{D}\ge0$ in \eqnref{eq:L_GKSL_D} being now time-independent.

\subsubsection{Semigroup-simulable (SS) dynamics}
\label{app:dyn_gens_fams_SS}
\begin{definition}
\label{def:sem_sim} 
We define a map $\Lambda_{t}$ to be (instantaneously) semigroup-simulable (SS) at time $t$ if its corresponding instantaneous exponent $\cZ_t=\log\Lambda_{t}$ in \eqnref{eq:inst_gens_exps} is of the GKSL form, i.e.:
\begin{align}
\cZ_t\!\left[\bullet\right] & =-\ii\left[\tilde{H}_{t},\bullet\right]+\tilde{\Phi}_{t}\!\left[\bullet\right]-\frac{1}{2}\left\{ \tilde{\Phi}_{t}^{\star}\!\left[\openone\right],\bullet\right\} ,\label{eq:Z_GKSL}
\end{align}
where, similarly to \eqnref{eq:L_GKSL}, $\tilde{H}_{t}$ and $\tilde{\Phi}_{t}$ are some Hermitian operator and CP map, respectively. If all instantaneous exponents, $\{\cZ_t\}_{t\ge0}$, can be decomposed according to \eqnref{eq:Z_GKSL}, we term the whole dynamics to be SS.
\end{definition}
Importantly, given that $\Lambda_{t}$ is SS at time $t$, a semigroup parametrised by $\tau\ge0$ with physical (of GKSL form) generator $\cL=\cZ_t$ may be defined:
\BE
\left\{ \tilde{\Lambda}(t)_{\tau}\right\} _{\tau\ge0}
\quad\text{with}\quad
\tilde{\Lambda}(t)_{\tau}=\ee^{\cZ_t\tau},
\label{eq:semi_sim}
\EE
so that it coincides with the original map at $\tau=1$, $\tilde{\Lambda}(t)_{\tau=1}=\Lambda_{t}$, or, in other words, ``simulates'' its action at this particular instance of ``fictitious time'' $\tau$.
\begin{obs}
\label{obs:ss_implies_phys}
The SS property provides a sufficient but not necessary condition for physicality of dynamics.
\end{obs}
If for a dynamical family $\left\{\Lambda_{t}=\ee^{\cZ_t}\right\}_{t\ge0}$ all its instantaneous exponents are of the GKSL form \eref{eq:Z_GKSL}, it must consist of maps which coincide with semigroups at all $t\ge0$ and, hence, all must be CPTP. In the other direction, however, there exist dynamics that are \emph{not} SS but nonetheless physical. Examples may be found by considering instances of, e.g.,random unitary and phase covariant, qubit evolutions, as shown below in \appref{app:qubit_classes}.

In case of \emph{commutative} dynamics, it follows from \eqnref{eq:map_via_gens_comm} that $\cZ_t=\int_{0}^{t}{\dd}\tau\cL_{\tau}$, so that one may explicitly connect the decomposition \eref{eq:L_GKSL_D} of the dynamical generator at time $t$ with the one of the instantaneous exponent in \eqnref{eq:Z_GKSL}, as follows
\AL{
\cZ_t\!\left[\bullet\right]=
&-\ii\left[\int_{0}^{t}\!\dd\tau H_{\tau},\bullet\right] \label{eq:Z_GKSL_int}\\
&+
\sum_{i,j}\int_{0}^{t}\!\dd\tau\,\msf{D}_{ij}(\tau)\left(F_{j}\bullet F_{i}^{\dagger}-\frac{1}{2}\left\{ F_{i}^{\dagger}F_{j},\bullet\right\} \right), \nonumber
}
where $\int_{0}^{t}{\dd}\tau H_{\tau}$ constitutes then $\tilde{H}_{t}$ in \eqnref{eq:Z_GKSL}.

Crucially, the decomposition \eref{eq:Z_GKSL_int} proves \lemref{lem:ss_cond_comm} of the main text, as it becomes clear that the GKSL form of $\cZ_t$ is then fully ensured by the condition 
\BE
\msf{\Gamma}(t):=\int_{0}^{t}\!{\dd}\tau\,\msf{D}(\tau)\ge0,
\label{eq:ss_cond_for_comm}
\EE
stated in \eqnref{eq:ss_cond_comm} of the main text. Hence, given a commutative dynamics for which condition \eref{eq:ss_cond_for_comm} holds, it must also be SS---constitute a \emph{semigroup-simulable and commutative} (SSC) evolution.

In some previous works \citep{Wolf2008,Cubitt2012}, the SS property has been identified as Markovianity of the dynamics. Let us emphasise that such a notion is non-trivially related to the concept of CP-divisibility introduced in \appref{app:dyn_gens_fams_CPdiv}, which is more commonly associated with Markovianity \citep{Rivas2014,Breuer2016,deVega2017}. The CP-divisibility ensures $\msf{D}(t)\ge0$ in \eqnref{eq:L_GKSL_D} at all times, so that (in case of commutative dynamics) the SSC condition \eref{eq:ss_cond_for_comm} is trivially fulfilled. However, as $\msf{D}(t)\ge0$ is a stronger requirement, there must exist (also commutative) evolutions that \emph{are} SS but \emph{not} CP-divisible, e.g., instances of qubit dynamics discussed below in \appref{app:dyn_gens_qubit_dyns}. This fact can also be understood by inspecting \eqnref{eq:back_trans}, from which it is clear that the GKSL form \eref{eq:Z_GKSL} of the instantaneous exponent $\cZ_t$ (and, hence, of $\cX_t=\dot{\mathcal{Z}}_{t}$) does not generally ensure the corresponding dynamical generator $\cL_t$ to also be of GKSL form \eref{eq:L_GKSL}.

\subsection{Rescalability of dynamical generators}
\label{app:dyn_gens_rescal}
\begin{definition}
\label{def:gen_rescal} 
We define a physical family of dynamical generators $\left\{ \cL_t\right\} _{t\ge0}$ to be rescalable if by multiplying all its elements by any non-negative constant, $\alpha\ge0$, one obtains a generator family,
\BE
\{\cL'_t:=\alpha\cL_t\}_{t\ge0},
\label{eq:L_rescaling}
\EE
that also yields physical dynamics. 
\end{definition}

Given a family of dynamical maps $\left\{ \Lambda_{t}\right\} _{t\ge0}$, by rescaling its corresponding dynamic generators $\left\{ \cL_t\right\} _{t\ge0}$, as in \eqnref{eq:L_rescaling}, we obtain a family of maps, $\left\{ \Lambda_{t}^{\prime}\right\} _{t\ge0}$, that according to \eqnref{eq:map_via_dyn_gens} reads
\BE
\Lambda_{t}^{\prime}=\mathcal{T}_{\leftarrow}\exp\!\left\{ \int_{0}^{t}\cL_{\tau}^{\prime}{\dd}\tau\right\} =\sum_{i=0}^{\infty}\mathcal{S}_{t}^{(i)}\!\left[\cL_{\bullet}^{\prime}\right]=\sum_{i=0}^{\infty}\alpha^{\ii}\mathcal{S}_{t}^{(i)}\!\left[\cL_{\bullet}\right].
\label{eq:dyson_alpha}
\EE
Crucially, as the above Dyson series includes now the factor $\alpha\ge0$, it is non-trivial to determine whether the resulting map $\Lambda_{t}^{\prime}$ is CPTP;~even in case of commutative dynamics for which time-ordering, $\mathcal{T}_\leftarrow$, can be dropped.

The rescalability, however, is naturally ensured in case of CP-divisible evolutions (and, hence, semigroups), as the GKSL form \eref{eq:L_GKSL} of any dynamical generator is then trivially carried over to $\cL'_t$ in \eqnref{eq:L_rescaling} for any $\alpha\ge0$. 

On the other hand, any family of dynamical generators yielding SSC dynamics must also be rescalable. As the instantaneous exponents are then related to the dynamical generators via $\cZ_t=\int_{0}^{t}{\dd}\tau\cL_{\tau}$, they transform similarly to \eqnref{eq:L_rescaling} with $\mathcal{Z}'_{t}:=\alpha\cZ_t$. Thus, the condition \eref{eq:ss_cond_for_comm} ensuring their GKSL form \eref{eq:Z_GKSL} is fulfilled for any $\alpha\ge0$.

Nevertheless, non-rescalability of generators also naturally emerges in some particular situations, e.g., when dealing with:

\paragraph{Dynamical generators with singularities,}\label{resc_sing}%
which emerge in case of evolutions whose family of dynamical maps, $\{\Lambda_t\}_{t\ge0}$, contains non-invertible CPTP maps. In this case, the dynamics can be unambiguously recovered from the dynamical generators only for times smaller than $T$, denoting the occurrence of the (first) singularity \cite{Andersson2007}. As a result, even though the dynamics is physical despite $\{\cL_t\}_{t\ge0}$ containing  singular generators, as soon as $\alpha\neq1$ in \eqnref{eq:L_rescaling} the integrability of the corresponding QME \eref{eq:master_equation}---and, hence, the physicality---is lost for times $t\ge T$. We provide an explicit example of such a phenomenon below in \appref{app:counter_gens_res}, where we discuss the Jaynes-Cummings model describing a qubit that undergoes spontaneous emission \cite{Breuer}.

\paragraph{Weak-coupling-based generators,}%
which are approximate and only valid for a particular timescale $T$ ($0\le t\le T$). 
Consider a family $\left\{ \lambda^{2}\cL_t\right\} _{0\le t\le T}$ of generators derived by employing a microscopic model and assuming the system-environment coupling constant, $\lambda$, small enough, so that the weak-coupling approximation to $O(\lambda^{2})$ holds and $\cL^{real}_t \approx \lambda^2\cL_t$ \citep{Gaspard1999} (e.g., by assuming the Redfield form of the QME \citep{Redfield1957}). One may then simply interpret the rescaling factor as the square of the coupling constant, $\alpha=\lambda^{2}$.  Importantly, such a generator family is guaranteed to yield physical dynamics---a family of CPTP maps---only on timescales with $T\ll\lambda^{-2}$ \citep{Dumcke1979}. Hence, by rescaling the generators with large enough $\alpha=\lambda^{2}$ or, in other words, by choosing strong enough coupling, one must at some point invalidate the weak-coupling approximation and, eventually, the physicality.

\paragraph{Commutative but not SS dynamics.}
Although all families of dynamical generators that lead to SSC dynamics must be rescalable, the commutativity property alone is not enough. 
A direct example is provided by the eternally non-Markovian model introduced in \refcite{Hall2014} and discussed below in \appref{app:qubit_classes}. In particular, when rescaling its generators according to \eqnref{eq:L_rescaling}, one obtains dynamics that is \emph{not} physical for any $0\le\alpha<1$ \citep{Benatti2017}.

\subsection{Additivity of dynamical generators}
\label{app:dyn_gens_add}
The \defref{def:gen_additivity} of the main text may be restated in a more detailed form as:
\begin{definition}
\label{def:gen_add}
Two families of physical and rescalable dynamical generators $\left\{ \cL_t^{(1)}\right\}_{t\ge0}$ and $\left\{ \cL_t^{(2)}\right\}_{t\ge0}$ are additive, if all their non-negative linear combinations,
\BE
\cL'_t:=\alpha\cL_t^{(1)}+\beta\cL_t^{(2)}
\label{eq:L'}
\EE
with $\alpha,\beta\ge0$, yield families of dynamical generators, $\left\{ \cL_t^{\prime}\right\}_{t\ge0}$, that are physical. 
\end{definition}

Firstly, we realise that (as for rescalability) all pairs of generator families describing CP-divisible evolutions must be additive, as by adding families of CP-divisible dynamics according to \eqnref{eq:L'} one obtains generators that are also of the GKSL form \eref{eq:L_GKSL}.

On the other hand, by considering generator families describing SSC dynamics, we observe that:
\begin{lem}
Any pair of SSC dynamics with generator families $\left\{ \cL_t^{(1)}\right\}_{t\ge0}$ and $\left\{ \cL_t^{(2)}\right\}_{t\ge0}$ whose addition \eref{eq:L'} yields commutative dynamics $\left\{ \cL_t^{\prime}\right\}_{t\ge0}$ for any $\alpha,\beta\ge0$ must be additive.
\label{lem:comm_ss_gens_are_add} 
\end{lem}
\begin{proof}
As all the families $\left\{ \cL_t^{(1)}\right\}_{t\ge0}$, $\left\{ \cL_t^{(2)}\right\}_{t\ge0}$ and $\left\{ \cL_t^{\prime}\right\}_{t\ge0}$ are commutative, their instantaneous exponents also add according to \eqnref{eq:L'}, i.e., $\cZ'_{t}=\alpha\cZ_t^{(1)}+\beta\cZ_t^{(2)}$. Moreover, as $\left\{ \cL_t^{(1)}\right\}_{t\ge0}$, $\left\{ \cL_t^{(2)}\right\}_{t\ge0}$ are SS, both $\cZ_t^{(1)}$ and $\cZ_t^{(2)}$ must satisfy \eqnref{eq:ss_cond_for_comm} with $\msf{\Gamma}^{(1)}(t)\ge0$ and $\msf{\Gamma}{}^{(2)}(t)\ge0$. Hence, any family $\left\{ \cL'_t\right\}_{t\ge0}$ must also be SS (and, hence, physical), as recomputing condition \eref{eq:ss_cond_for_comm} for $\cZ'_{t}$ with help of \eqnref{eq:Z_GKSL_int} it reads
\BE
\msf{\Gamma}'(t)=\alpha\,\msf{\Gamma}^{(1)}(t)+\beta\,\msf{\Gamma}^{(2)}(t)\ge0,
\EE
and is trivially fulfilled for any $\alpha,\beta\ge0$.
\end{proof}
%

\section{Rescalability and additivity of qubit dynamical generators}
\label{app:dyn_gens_qubit_dyns}

\subsection{RU and PC classes of qubit dynamics}
\label{app:qubit_classes}
We consider two important classes of commutative qubit dynamics, namely, \emph{random unitary} (RU) \cite{Andersson2007,Chruscinski2013} and \emph{phase-covariant} (PC) \cite{Smirne2016} evolutions. In order to provide their physically motivated instances, we explicitly discuss exemplary microscopic derivations for the (generalised) \emph{dephasing} and \emph{amplitude damping} models that fall into the RU and PC classes, respectively.

\subsubsection{Random unitary (RU) dynamics}
\label{app:dyn_gens_qubit_dyns_RU}
Random unitary (RU) qubit dynamics are formed by considering smooth families of Pauli channels, which up to unitary transformations represent the most general qubit unital ($\Lambda[\openone]=\openone$) maps \cite{Vacchini2012}. In particular, any RU evolution is described by a qubit QME \eref{eq:master_equation} with \citep{Andersson2007,Chruscinski2013}:
\BE
\cL_t\!\left[\bullet\right]=\sum_{k=\{x,y,z\}}\gamma_{k}(t)\left(\hs_{k}\bullet\hs_{k}-\bullet\right),
\label{eq:dyn_gens_RU}
\EE
so that the RU generator family $\left\{ \cL_t\right\} _{t\ge0}$ is fully specified by the three rates $\gamma_{k}(t)$ defining a diagonal form of the general $\msf{D}$-matrix in \eqnref{eq:L_GKSL_D}, with Pauli operators $\left\{ \openone,\sx,\sy,\sz\right\}$ constituting a basis for two-dimensional Hermitian matrices. Hence, it directly follows that the GKSL form \eref{eq:L_GKSL} of RU generators, and hence the \emph{CP-divisibility} \eref{eq:CP_div} of the dynamics, is ensured iff at all times all $\gamma_{k}(t)\ge0$ are non-negative in \eqnref{eq:dyn_gens_RU}. 

One may straightforwardly verify that any RU dynamics \eref{eq:dyn_gens_RU} is \emph{commutative}, so that by \eqnref{eq:inst_dyn_gens_coinc_for_comm_dyn} dynamical and instantenuous generators coincide, and the instantaneous exponents according to \eqnref{eq:Z_GKSL_int} read:
\BE
\cZ_t\!\left[\bullet\right] =\sum_{k=\{x,y,z\}}\Gamma_{k}(t)\left(\hs_{k}\bullet\hs_{k}-\bullet\right)
\label{eq:Z_t_rand_unit}
\EE
with $\Gamma_{k}(t):=\int_{0}^{t}{\dd}\tau\gamma_{k}(\tau)$. Hence, it directly follows from the condition \eref{eq:ss_cond_for_comm} that any RU dynamics is \emph{SS} iff
\BE
\forall_{t\ge0,k=\{x,y,z\}}:\quad\Gamma_{k}(t)\ge0.
\label{eq:ss_rand_unit}
\EE
Note that, as one may easily construct families of dynamical generators \eref{eq:dyn_gens_RU} that satisfy \eqnref{eq:ss_rand_unit} without requiring $\gamma_{k}(t)\ge0$ for all $t$, there exist RU dynamics that are SS but \emph{not} CP-divisible.

On the other hand, an explicit condition for the physicality of RU dynamics is known \citep{Andersson2007,Chruscinski2015}. In particular, RU dynamics is physical iff for all the cyclic permutations of $i,j,k\in\{x,y,z\}$ (i.e., such that $\epsilon_{ijk}=1$):
\BE
\mu_{i}(t)+\mu_{j}(t)\le1+\mu_{k}(t),
\label{eq:rand_unit_phys}
\EE
where each $\mu_{i}(t):=\exp\!\left[-2\left(\Gamma_{j}(t)+\Gamma_{k}(t)\right)\right]$. It is easy to verify that the physicality condition \eref{eq:rand_unit_phys} is less restrictive than the SS condition \eref{eq:ss_rand_unit}. Hence, there exist RU dynamics that are physical but \emph{not} SS, despite being commutative.

\paragraph{Eternally non-Markovian model.}
An example of RU dynamics that is physical but \emph{not} SS is also provided by the \emph{eternally non-Markovian} model introduced in \refcite{Hall2014}, which corresponds to the following choice of rates in \eqnref{eq:dyn_gens_RU}:
\BE
\gamma_{x}(t)=\gamma_{y}(t)=\frac{1}{2},\quad\gamma_{z}(t)=-\frac{1}{2}\tanh\!\left(t\right),
\label{eq:gammas_eternalNM}
\EE
for which the physicality condition \eref{eq:rand_unit_phys} holds, even though $\gamma_{z}(t)<0$ (and hence $\Gamma_{z}(t)<0$) for all $t\ge0$.

\paragraph{Dephasing dynamics.}
The simplest example of RU dynamics \eref{eq:dyn_gens_RU} is provided by the \emph{dephasing} model:
\BE
\cL_t\!\left[\bullet\right]=\gamma(t)\left(\hs_{\boldsymbol{n}}\bullet\hs_{\boldsymbol{n}}-\bullet\right),
\label{eq:dyn_gens_dephasing}
\EE
where $\hs_{\boldsymbol{n}}=\boldsymbol{n}\cdot \sV=\sum_{i}n_{i}\hs_{i}$, and $\hs_{\boldsymbol{n}}^{2}=\openone$ implies $\left\Vert \boldsymbol{n}\right\Vert =1$. The unit vector $\boldsymbol{n}$ should be interpreted as a choice (a passive rotation in the Bloch-ball picture) of the Pauli-operator basis, in which then \eqnref{eq:dyn_gens_dephasing} corresponds to (rank-one Pauli) RU dynamics \eref{eq:dyn_gens_RU} with only a single term present in the sum. One may easily verify that for the dephasing model to be physical $\Gamma(t)=\int_{0}^{t}{\dd}\tau\gamma(\tau)\ge0$, with the notions of physicality and SS then trivially coinciding.

The dephasing dynamics \eref{eq:dyn_gens_dephasing} can be explicitly obtained by considering various microscopic derivations, in which a qubit is coupled to a large environment via some $H_{\mathrm{int}}\propto \hs_{\boldsymbol{n}}\otimes O_\mathrm{env}$. In \secref{sec:magnet_dephasing} of the main text, we provide a compact example by using a toy-model of a qubit coupled to a large magnet. The most common microscopic derivation, however, is constructed by considering a qubit coupled to a large, thermal bosonic bath \cite{Breuer}. The interaction is then modelled by $H_{\mathrm{int}}\propto\hs_{\boldsymbol{n}}\otimes\left(\sum_{k}g_{k}\hat a_{k}+g_{k}^{\star}\hat a_{k}^{\dagger}\right)$ which couples the qubit to a bosonic reservoir of an Ohmic-like spectral density \citep{Leggett1987}:
\BE
J(\omega):=\sum_{k}g_{k}^{2}\delta(\omega-\omega_{k})=\frac{\omega^{s}}{\omega_{c}^{s-1}}\,\ee^{-\frac{\omega}{\omega_{c}}},
\EE
where $\omega_{c}$ represents the reservoir cutoff frequency, while $s\ge0$ is the so-called Ohmicity parameter. 

Assuming further the reservoir to be at zero temperature, the dynamical generators describing the qubit evolution take then exactly the form \eref{eq:dyn_gens_dephasing} with the dephasing rate reading \citep{Haikka2013}:
\BE
\gamma(t)=\omega_{c}\left[1-\left(\omega_{c}t\right)^{2}\right]^{-\frac{s}{2}}\Gamma[s]\,\sin\!\left[s\,\arctan\!\left(\omega_c t\right)\right],
\label{eq:deph_rate_Ohmic}
\EE
where $\Gamma[s]$ above represents the Euler gamma function. Moreover, one may show that the dephasing rate temporarily takes negative values iff $s>2$, so that the dynamics ceases then to be CP-divisible \citep{Haikka2013}.

\subsubsection{Phase-covariant (PC) dynamics}
A phase-covariant (PC) qubit evolution corresponds to a family of dynamical maps that possess azimuthal symmetry with respect to rotations about the $z$ axis in the Bloch-ball representation. The most general PC dynamics is described by a qubit QME \eref{eq:master_equation} with \citep{Smirne2016}:
\begin{align}
\cL_t\!\left[\bullet\right] =&\quad\gamma_{-}(t)\left(\sm\bullet\sp-\frac{1}{2}\left\{ \sp\sm,\bullet\right\} \right) \nonumber\\
&+\gamma_{+}(t)\left(\sp\bullet\sm-\frac{1}{2}\left\{ \sp\sm,\bullet\right\} \right) \nonumber\\
&+\gamma_{z}(t)\bigg(\sz\bullet\sz-\bullet\bigg), 
\label{eq:dyn_gens_PC}
\end{align}
which represents a combination of relaxation, excitation and dephasing processes occurring with rates:~$\gamma_{-}(t)$, $\gamma_{+}(t)$ and $\gamma_{z}(t)$, respectively; while $\hs_{\pm}:=\frac{1}{\sqrt{2}}\left(\sx\pm\ii\sy\right)$ are the transition operators. 

Although dynamical generators commute within each of the RU \eref{eq:dyn_gens_RU} and PC \eref{eq:dyn_gens_PC} classes of dynamics, they do \emph{not} generally commute in between the two. In fact, their common commutative subset corresponds to \emph{all unital PC evolutions} for which $\cL_t$ in \eqnsref{eq:dyn_gens_RU}{eq:dyn_gens_PC} coincide with $\gamma_x(t)=\gamma_y(t)=\frac{1}{2}\gamma_+(t)=\frac{1}{2}\gamma_-(t)$. Note that, the eternally non-Markovian model with decay rates specified in \eqnref{eq:gammas_eternalNM} is, in fact, both RU and PC, while the dephasing (RU) dynamics belongs to the PC class only when aligned along the $z$ direction, i.e., when $\boldsymbol{n}=\{0,0,1\}$ in \eqnref{eq:dyn_gens_dephasing}.

As $\left\{ \openone,\sp,\sm,\sz\right\}$ equivalently constitute a basis for two-dimensional Hermitian matrices, the properties of PC dynamics can be determined analogously to the RU case. In particular, considering now $k=\{+,-,z\}$, a given PC evolution is \emph{CP-divisible} iff all $\gamma_{k}(t)\ge0$ at all times, while it is \emph{SS} iff all $\Gamma_{k}(t)=\int_{0}^{t}\gamma_{k}(\tau){\dd}\tau\ge0$
for any $t\ge0$. Moreover, it is not hard to verify that the family of PC generators \eref{eq:dyn_gens_PC} is \emph{physical} iff for all $t\ge0$ \citep{Smirne2016}:
\BE
\eta_{||}(t)\pm\kappa(t)\le1\quad\mathrm{and}\quad\left(1+\eta_{||}(t)\right)^{2}\ge4\eta_{\perp}(t)^{2}+\kappa(t)^{2},
\label{eq:phys_conds_PC_dyns}
\EE
where $\eta_{||}(t)\!:=\!\ee^{-\delta(t)}$, $\eta_{\perp}(t)\!:=\!\ee^{-\frac{1}{2}\left(\delta(t)-4\Gamma_{z}(t)\right)}$, $\kappa(t)\!:=\!\ee^{-\delta(t)}\int_{0}^{t}\!\dd\tau\,\ee^{\delta(\tau)}\!\left[\gamma_{+}(\tau)-\gamma_{-}(\tau)\right]$
and $\delta(t)\!:=\!\Gamma_{+}(t)+\Gamma_{-}(t)$.

Hence, similarly to the case of RU dynamics, as the physicality condition \eref{eq:phys_conds_PC_dyns} is less restrictive than the SS requirement ($\Gamma_k(t)\ge0$), there exist (commutative) PC dynamics that are \emph{not} SS but still physically legitimate. The eternally non-Markovian model \cite{Hall2014}, being both RU and PC, provides again an appropriate example.

\paragraph{Amplitude damping dynamics.}
The most common example of PC dynamics, which is not RU, is the \emph{amplitude damping} evolution that represents the pure relaxation process, i.e., spontaneous emission of a two-level (qubit) system \citep{Breuer};~and corresponds to the choice $\gamma_{+}(t)=\gamma_{z}(t)=0$ in \eqnref{eq:dyn_gens_PC}, i.e.,
\BE
\cL_t\!\left[\bullet\right] = \gamma_{-}(t)\left(\sm\bullet\sp-\frac{1}{2}\left\{ \sp\sm,\bullet\right\} \right). 
\label{eq:dyn_gens_ampl_damp}
\EE

Its canonical microscopic derivation stems from the Jaynes-Cummings interaction model \citep{Breuer}, $H_{\mathrm{int}}\propto\sp\otimes\left(\sum_{k}g_{k}\hat a_{k}+g_{k}^{\star}\hat a_{k}^{\dagger}\right)$, in which the qubit is coupled to a cavity possessing Lorentzian spectral density \citep{Li2010}:
\BE
J(\omega)=\frac{1}{2\pi}\frac{\gamma_{0}\lambda^{2}}{(\omega_{0}-\omega-\Delta)^{2}+\lambda^{2}},
\EE
where $\Delta$ describes the difference between qubit transition, $\omega_{0}$, and cavity central frequencies, while $\lambda$ represents the cavity spectral width. Crucially, such a model---after tracing out degrees of freedom of the cavity---leads to a qubit QME \eref{eq:master_equation} with the dynamical generator \eref{eq:dyn_gens_ampl_damp}, whose relaxation rate reads \citep{Li2010}:
\BE
\gamma_{-}(t)=\mathrm{Re}\!\left\{ \frac{2\gamma_{0}\lambda}{\lambda-\ii\Delta+d\,\coth\left(\frac{dt}{2}\right)}\right\},
\label{eq:damp_JC}
\EE
and does not exhibit a singular behaviour as long as the real part of the complex parameter $d:=\sqrt{(\lambda-\ii\Delta)^{2}-2\gamma_{0}\lambda}$ is positive. 

However, this is not the case in the \emph{on-resonance} ($\Delta=0$), \emph{strong-coupling} ($\gamma_{0}>\frac{\lambda}{2}$) regime, in which $d$ becomes purely imaginary, $d=\ii\left|d\right|$ with $\left|d\right|=\sqrt{2\gamma_{0}\lambda-\lambda^{2}}$, so that the relaxation rate \eref{eq:damp_JC} simplifies to
\BE
\gamma_{-}(t)=\frac{2\gamma_{0}\lambda}{\lambda+\left|d\right|\cot\!\left(\frac{\left|d\right|t}{2}\right)},
\label{eq:damp_JC_div}
\EE
and diverges at every $t=\frac{2}{\left|d\right|}\!\left(\mathrm{arccot}\!\left(\frac{-\lambda}{\left|d\right|}\right)+n\pi\right)$ with $n\in\mathbb{N}^{+}$ \citep{Breuer}.

\subsection{Counterexamples to rescalability and additivity}
\label{app:counter_qubit}

\subsubsection{Non-rescalable qubit dynamics}
\label{app:counter_gens_res}
An explicit example of \emph{non-rescalable} qubit dynamics is provided by the Jaynes-Cummings model of spontaneous emission described just above, considered in the \emph{on-resonance, strong-coupling} regime. It leads to an example of dynamics described in \appref{resc_sing} with dynamical generators being singular due to the damping rate \eref{eq:damp_JC_div} being divergent periodically in $t$. 

Considering then the rescaled version of the amplitude-damping generator \eref{eq:dyn_gens_ampl_damp}, i.e., $\cL'_t$ of \eqnref{eq:L_rescaling}, one may simply integrate the resulting QME \eref{eq:master_equation} for any $t$ in order to explicitly determine the form of the rescaled dynamical map $\Lambda'_t$ in \eqnref{eq:dyson_alpha}. For instance, when setting $\gamma_{0}=3/2$ and $\lambda=1$ for simplicity in \eqnref{eq:damp_JC_div}, the corresponding CJ-matrix \eref{eq:CJ_matrix}, $\Omega_{\Lambda'_{t}}$, may be explicitly computed and its non-zero eigenvalues read 
\BE
\lambda_{\pm}^{\mathrm{vals}}=1\pm2^{-\alpha}\ee^{-\alpha t}\left[\sqrt{2}\cos\!\left(\frac{t}{\sqrt{2}}\right)+\sin\!\left(\frac{t}{\sqrt{2}}\right)\right]^{2\alpha}.
\label{eq:evals}
\EE
Crucially, although in the case of original dynamics (when $\alpha=1$) both $\lambda_{\pm}^{\mathrm{vals}}\ge0$ for any $t\ge0$, \emph{only} for times before the occurrence of the first singularity, i.e., when $t < T=\frac{2}{\left|d\right|}\left(\mathrm{arccot}\!\left(\frac{-\lambda}{\left|d\right|}\right)+\pi\right)=\sqrt{2}\left[\pi-\arctan(\sqrt{2})\right]$, the eigenvalues are guaranteed to be real and non-negative independently of $\alpha$. In particular, for any $t \ge T$ one may easily find $\alpha\ge0$ ($\alpha\ne1$) such that the eigenvalues \eref{eq:evals} take complex values with the QME being, in fact, not even integrable.

\subsubsection{Non-additive qubit dynamics}
\label{app:counter_gens_add}
In order to construct counterexamples to additivity, we consider dynamical generators given in \eqnref{eq:ex_gens} of the main text and make specific choices for their dissipation rates $\gamma_1(t)$ and $\gamma_2(t)$. We then solve the QME obtained after adding the generators as in \eqnref{eq:L'} with some $\alpha,\beta \geq 0$, i.e.,
\BE
\label{eq.summedmastereq}
\frac{d}{dt} \rho(t) = \cL_t'[\rho(t)] = \alpha \cL^{(1)}_t[\rho(t)] + \beta \cL^{(2)}_t[\rho(t)],
\EE
in order to explicitly compute the corresponding family of maps $\{\Lambda'_t\}_{t\ge0}$. Crucially, we find in this way families containing maps that cease to be CPTP---with their CJ matrices, $\{\Omega_{\Lambda'_t}\}_{t\ge0}$ as defined in \eqnref{eq:CJ_matrix}, exhibiting negative eigenvalues at some time instances.

In order to solve the QME \eref{eq.summedmastereq}, we choose a qubit operator basis:
\BE
\hat{\mu}_0 = \1/\sqrt{2},\; \hat{\mu}_1 = \sx/\sqrt{2},\; \hat{\mu}_2 = \sy/\sqrt{2},\; \hat{\mu}_3 = \sz/\sqrt{2},
\EE
which allows us to use the matrix and vector representations for generators and states, respectively. As $\tr[\hat{\mu}_i \hat{\mu}_j] = \delta_{ij}$, any generator $\cL$ may then be represented by a matrix $\mat{M}$ with entries $M_{ij} = \tr[\hat{\mu}_i \cL[\hat{\mu}_j]]$, while any state $\rho$ by a vector $\vec{x}$ with components $x_i = \tr[\rho \hat{\mu}_i]$ ($x_0 = \tr[\rho]/\sqrt{2}=\sqrt{2}/2$ by definition). The QME \eref{eq.summedmastereq} is then equivalent to the set of linear, coupled differential equations:
\BE
\frac{d}{dt} \vec{x}(t)  = \mat{M}_t'\,\vec{x}(t),
\label{eq:eom_add_gens}
\EE
where $\mat{M}_t'$ is the matrix representation of $\cL_t'$.

\begin{figure*}[t!]
\begin{center}
\includegraphics[width=0.8\linewidth]{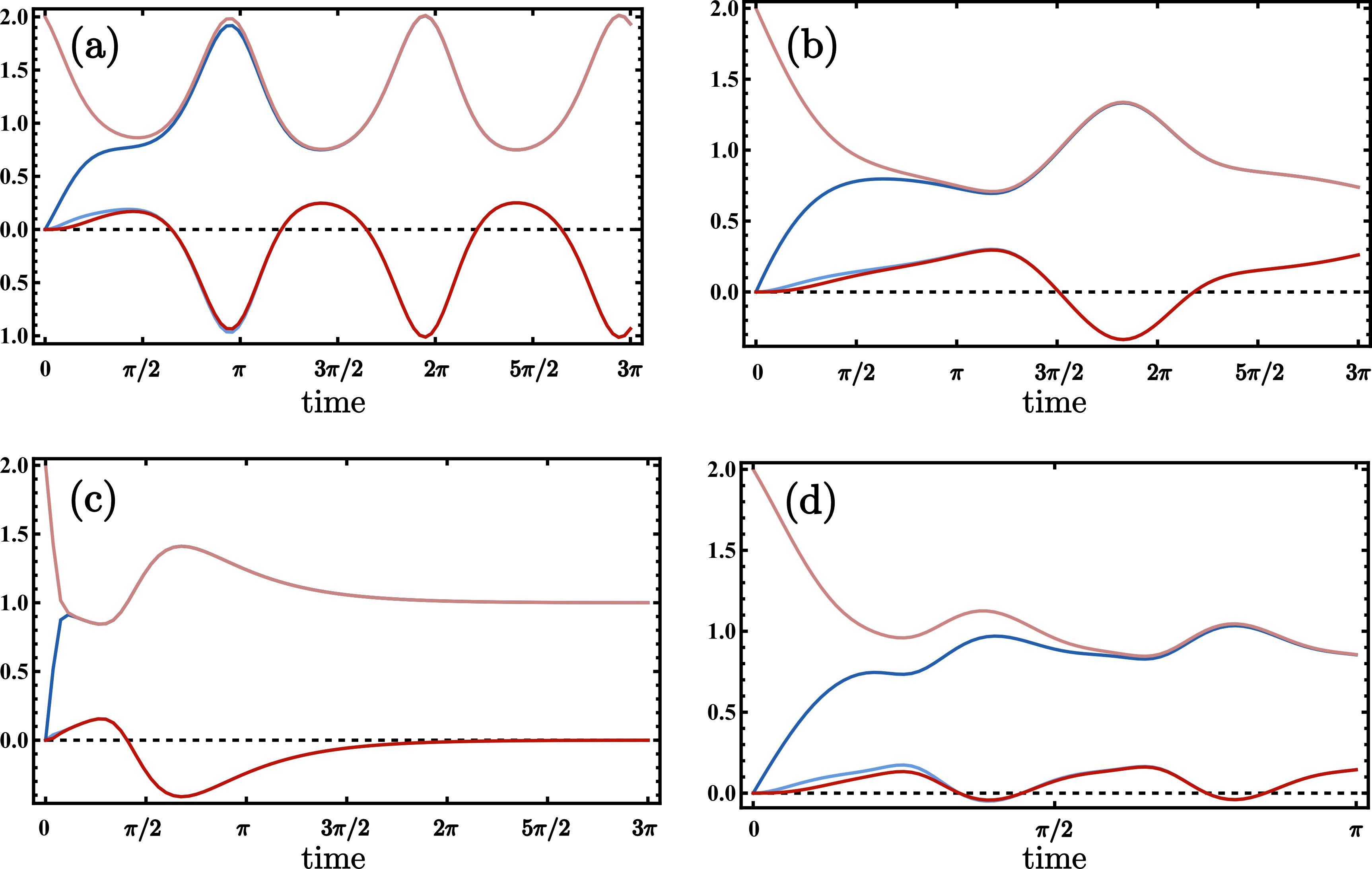}
\caption{%
\textbf{CJ eigenvalues as functions of time after adding generators}, $\alpha\cL_t^{(1)}+\beta\cL_t^{(2)}$. $\cL_t^{(1)}$ and $\cL_t^{(2)}$ describe qubit dephasing along $x$ and amplitude damping in $z$, respectively, as in \eqnref{eq:ex_gens} of the main text. In all plots $\alpha=\beta=1$, while the rate functions are chosen so that:~%
(a) $\gamma_1(t) = \sin(2t)$, while $\gamma_2(t) = 1$; 
(b) $\gamma_1(t) = 1/2$, while $\gamma_2(t) = \sin(t)$;
(c) $\gamma_1(t)$ is set according to \eqnref{eq:deph_rate_Ohmic} (super-Ohmic regime) with cut-off frequency $\omega_c=1$ and Ohmicity parameter $s=4.5$, while $\gamma_2(t) = 1$;
(d) $\gamma_1(t) = 1$, while $\gamma_2(t)$ is fixed according to \eqnref{eq:damp_JC} (off-resonant regime) with detuning $\Delta=3$, spectral width $\lambda=0.05$ and excited-state decay rate $\gamma_0=150$. \emph{Note that in all cases negative eigenvalues occur, indicating that each evolution ceases to be physical at some point in time.}
}
\label{fig:evalsunphys}
\end{center}
\end{figure*}

For our first example, we take dephasing and amplitude damping generators of \eqnref{eq:ex_gens} in the main text to be non-Markovian and semigroup, respectively, with
\BE
\gamma_1(t) = \sin(\omega t)
\qquad \text{and} \qquad
\gamma_2(t) = \gamma,
\label{eq:nonMarko_deph}
\EE
where $\omega,\gamma > 0$ are some fixed constants. Crucially, since for all $t\ge0$:
\BE
\int_0^t \dd s \,\gamma_1(s) = \frac{1 - \cos(\omega t)}{\omega}\; \geq\; 0,
\EE  
and similarly in the case of semigroup $\gamma_2(t)$, both generator families are SSC and hence rescalable, so that their additivity may be unambiguously considered.

Considering their non-negative linear combinations, $\cL_t' = \alpha \cL^{(1)}_t + \beta \cL^{(2)}_t$, one obtains generator families with
\BE
\mat{M}'_t =
\begin{pmatrix}
 0 & 0 & 0 & 0 \\
 0 & -\frac{\beta  \gamma }{2} & 0 & 0 \\
 0 & 0 & -\frac{\beta  \gamma }{2}-2 \alpha  \sin (\omega t ) & 0 \\
 \beta  \gamma  & 0 & 0 & -\beta  \gamma -2 \alpha  \sin ( \omega t) \\
\end{pmatrix}.
\EE
Solving \eqnref{eq:eom_add_gens}, one finds
\BE
\vec{x}(t) =
\begin{pmatrix}
 x_0(0) \\
 \ee^{-\frac{1}{2} \beta  \gamma  t} x_1(0)  \\
 \ee^{-\frac{2 \alpha }{\omega }+\frac{2 \alpha  \cos (\omega t )}{\omega }-\frac{\beta  \gamma  t}{2}} x_2(0)  \\
 \ee^{-\frac{2 \alpha }{\omega }+\frac{2 \alpha  \cos (\omega t )}{\omega }-\beta  \gamma  t} [ x_3(0) + \beta  \gamma  \ee^{\frac{2 \alpha }{\omega }} I(t)\, x_0(0) ]
\end{pmatrix},
\EE
where $x_0(0)=\sqrt{2}/2$ and
\BE
I(t) = \int_0^t \!\dd s\,\ee^{\beta  \gamma  s-\frac{2 \alpha  \cos (\omega s)}{\omega }}.
\label{eq:integral1}
\EE
The four eigenvalues of the CJ matrices, $\Omega_{\Lambda'_t}$, for the corresponding family of maps $\{\Lambda'_t\}_{t\ge0}$ read
\begin{widetext}
\BE
\lambda^\trm{vals}_{\mp,\pm}=
\frac{1}{2} \ee^{-\frac{2 \alpha }{\omega }-\beta  \gamma  t}
\left(
\ee^{\frac{2 \alpha }{\omega }+\beta  \gamma  t} \mp \ee^{\frac{2 \alpha  \cos (\omega t )}{\omega }} \pm \sqrt{\beta ^2 \gamma ^2 \ee^{\frac{4 \alpha  (\cos (\omega t )+1)}{\omega }} I(t)^2\mp\ee^{\beta  \gamma  t} \left(\ee^{\frac{2 \alpha }{\omega }} + \ee^{\frac{2 \alpha  \cos (\omega t )}{\omega }}\right)^2}
\right),
\EE
\end{widetext}
and are plotted in \figref{fig:evalsunphys}(a) for $\alpha=\beta=1$, $\omega=2$, and $\gamma=1$. For $t=\pi$, the integral in \eqnref{eq:integral1} evaluates to $I\approx 23.36$, and it is easy to check that two of the eigenvalues are negative. Hence, the evolution is clearly unphysical.

For the second example, we consider the symmetric case with the dissipation rates exchanged, i.e.: 
\BE
\gamma_1(t) = \gamma 
\qquad \text{and} \qquad
\gamma_2(t)  = \sin(\omega t).
\label{eq:nonMarko_amp}
\EE
By the same argumentation as before both generators are SSC (and thus rescalable) and upon addition yield
\BE
\mat{M}'_t = \sin (\omega t)
\begin{pmatrix}
 0 & 0 & 0 & 0 \\
 0 & -\frac{\beta}{2} & 0 & 0 \\
 0 & 0 & -\frac{2 \alpha\gamma}{\sin (\omega t)} -\frac{\beta}{2}\beta & 0 \\
 \beta & 0 & 0 & -\frac{2\alpha\gamma}{\sin (\omega t)} -\beta \\
\end{pmatrix},
\EE
which after solving \eqnref{eq:eom_add_gens} leads to
\BE
\mathbf{x}(t) =
\begin{pmatrix}
 x_0(0) \\
 \ee^{\frac{\beta  \cos (\omega t )}{2 \omega }-\frac{\beta }{2 \omega }} x_1(0)  \\
 \ee^{-\frac{\beta }{2 \omega }-2 \alpha  \gamma  t+\frac{\beta  \cos (\omega t )}{2 \omega }} x_2(0)  \\
 \ee^{-\frac{\beta }{\omega }-2 \alpha  \gamma  t+\frac{\beta  \cos (\omega t )}{\omega }} [x_3(0) + \beta  \ee^{\beta /\omega } I(t)\, x_0(0)]
\end{pmatrix}
\EE
with again $x_0(0)=\sqrt{2}/2$ and now
\BE
I(t) = \int_0^t \!\dd s\, \sin (\omega s)\, \ee^{2 \alpha  \gamma  s-\frac{\beta  \cos (\omega s)}{\omega }}.
\label{eq:integral2}
\EE
The four CJ eigenvalues this time read:
\begin{widetext}
\BE
\lambda^\trm{vals}_{\mp,\pm}=\frac{1}{2} \ee^{-\frac{\beta }{\omega }-2 \alpha  \gamma  t} \left(\ee^{\frac{\beta }{\omega }+2 \alpha  \gamma  t}\mp\ee^{\frac{\beta  \cos (\omega t )}{\omega }} \pm \sqrt{e^{\frac{\beta  (\cos (\omega t )+1)}{\omega }} \left(\beta ^2 \ee^{\frac{\beta  (\cos (\omega t )+1)}{\omega }} I(t)^2+\left(\ee^{2 \alpha  \gamma  t}+1\right)^2\right)}\right) ,
\EE
\end{widetext}
and are plotted in \figref{fig:evalsunphys}(b) for $\alpha=\beta=1$, $\omega=1$, and $\gamma=1/2$. For $t=2\pi$ the integral \eref{eq:integral2} yields $I\approx -204.81$, and again two of the eigenvalues are negative, proving  the evolution to be unphysical.

We repeat the above analysis, but considering this time the dissipation rates of either dephasing or amplitude-damping in \eqnref{eq:ex_gens} of the main text to have a functional form derived explicitly from an underlying microscopic model yielding non-Markovian dynamics.

Firstly, in an analogy to \eqnref{eq:nonMarko_deph}, we consider the dephasing rate to be specified by \eqnref{eq:deph_rate_Ohmic}---as if the qubit were coupled to a \emph{bosonic reservoir with an Ohmic-like spectrum}---while the damping rate to be constant. In this case, we can solve \eqnref{eq:eom_add_gens} numerically at each $t$ for given parameter settings, in order to compute the corresponding CJ eigenvalues. These are plotted in \figref{fig:evalsunphys}(c) for $\alpha=\beta=1$, $\omega_c=1$, $s=4.5$, which corresponds to a super-Ohmic spectrum \cite{Haikka2013}, and $\gamma=1$. We observe that the evolution becomes unphysical around $t=\pi/2$.

Secondly, we consider the symmetric case in an analogy to \eqnref{eq:nonMarko_amp}, this time setting the dephasing to be constant, while the damping rate to the one of \eqnref{eq:damp_JC}---derived basing on the \emph{Jaynes-Cummings microscopic model} in which the qubit is coupled to a cavity with a Lorentzian frequency spectrum---in the off-resonant ($\Delta\ne0$) regime. Again, we find the CJ eigenvalues by solving \eqnref{eq:eom_add_gens} numerically for fixed parameter values. These are plotted in \figref{fig:evalsunphys}(d) for $\alpha=\beta=1$, $\Delta=3$, $\lambda=0.05$, $\gamma_0=150$, and $\gamma=1$. We observe again that the evolution becomes unphysical, this time a bit before $t=\pi/2$.

\section{Microscopic derivations of QMEs}
\label{app:QMEs_micro}
%

We consider the situation depicted in \figref{fig.settings}(a) of the main text, in which a system of interest and its environment evolve under closed dynamics determined by a time-invariant \emph{total} (T) Hamiltonian---consisting of Hamiltonians associated with the \emph{system} (S), the \emph{environment} (E) and their \emph{interaction} (I):
\BE
H_{T}=H_{S}+H_{E}+H_{I}.
\label{eq:H_T}
\EE

\subsection{Interaction and Schr\"odinger pictures}
\label{app:SP_and_IP}
The \emph{interaction picture} (IP), which we denote here with an over-bar, is then defined in the same manner for all operators and states acting on the system-environment Hilbert space, i.e., as
\begin{align}
\bar{O} & :=\ee^{\ii\left(H_{S}+H_{E}\right)t}\,O\,\ee^{-\ii\left(H_{S}+H_{E}\right)t}
\label{eq:IP_def}
\end{align}
for any given $O\in\cB(\cH_S\otimes\cH_E)$ that is specified in the \emph{Schr\"odinger picture} (SP).

In contrast, a general dynamical map, $\Lambda_{t,t_0}$, that describes the evolution of solely the system between the initial time $t_0$ and some later $t$ transforms from SP to IP (and \emph{vice versa}) as:
\BE
\bar{\Lambda}_{t,t_0} = \cU_t^{S\dagger}\,\circ\,\Lambda_{t,t_0}\,\circ\,\cU_{t_0}^{S}
\quad\Big(\!\Longleftrightarrow\,
\Lambda_{t,t_0}\! = \cU_t^{S}\,\circ\,\bar{\Lambda}_{t,t_0}\,\circ\,\cU_{t_0}^{S\dagger}\Big),
\label{eq:Lambda_ts_IP_SP}
\EE
where by 
\BE
\cU_t^{S}[\,\bullet\,]:=U_S(t)\,\bullet\,U_S^\dagger(t)=\ee^{-\ii H_{S}t}\,\bullet\, \ee^{\ii H_{S}t}
\label{eq:U_S_map}
\EE
we denote the unitary transformation induced by the system free Hamiltonian, $H_S$. However, as we consider throughout this work dynamical maps that commence at zero time ($t_0=0$), see \eqnref{eq:dynamics}, \eqnref{eq:Lambda_ts_IP_SP} simplifies to 
\BE
\bar{\Lambda}_t= \cU_t^{S\dagger}\circ\Lambda_t
\quad\Big(\!\Longleftrightarrow\quad
\Lambda_t= \cU_t^{S}\circ\bar{\Lambda}_t\Big), 
\label{eq:Lambda_t_IP_SP}
\EE
which allows us to explicitly compute how the corresponding dynamical generators of $\Lambda_t$ and $\bar{\Lambda}_t$ transform between the SP and IP. 

In particular, defining the \emph{IP-based dynamical generator} in accordance with \eqnref{eq:dyn_gen_def} as
\BE
\LIP_t :=\dot{\bar{\Lambda}}_{t}\circ\bar{\Lambda}_{t}^{-1},
\label{eq:L_IP}
\EE
and substituting for $\bar{\Lambda}_t$ according to \eref{eq:Lambda_t_IP_SP}, we obtain
\begin{align}
\LIP_t
&=
(\dot{\cU}_t^{S\dagger}\circ\Lambda_t+\cU_t^{S\dagger}\circ\dot{\Lambda}_t)\circ\Lambda_t^{-1}\circ\cU_t^{S} \\
&=
\dot{\cU}_t^{S\dagger}\circ\cU_t^{S}+\cU_t^{S\dagger}\circ\dot{\Lambda}_t\circ\Lambda_t^{-1}\circ\cU_t^{S} \\
&=
\ii[H_S,\bullet]+\cU_t^{S\dagger}\circ\LSP_t\circ\cU_t^{S},
\label{eq:L_t_IP_SP}
\end{align}
where in the last line we have used the definition of $\cU_t^S$ \eref{eq:U_S_map}, and accordingly defined the \emph{SP-based dynamical generator}, i.e., as in \eqnref{eq:dyn_gen_def}: 
\BE
\LSP_t :=\dot{\Lambda}_{t}\circ\Lambda_{t}^{-1},
\label{eq:L_SP}
\EE

\subsection{QME in the integro-differential form}
\label{app:QME_integrodiff}
The von Neumann equation describing the unitary evolution of the closed system-enviroment (SE) system, i.e., \eqnref{eq:ODE_S+E} of the main text, in the IP reads:
\BE
\frac{d\bar{\rho}_{SE}(t)}{dt}=-\ii\left[\bar{H}_{I}(t),\bar{\rho}_{SE}(t)\right].
\label{eq:vN}
\EE
Assuming the SE to initially be in a product state,
\BE
\rhoSE(0)=\rho_{S}(0)\otimes\rho_{E}
\label{eq:init_prod}
\EE
with $\rho_{E}$ being a \emph{stationary} state of the environment that satisfies $\left[\bar{H}_{E}(t),\rho_{E}\right]=\left[H_{E},\rho_{E}\right]=0$, one may write the integral of \eqnref{eq:vN} as:
\BE
\bar{\rho}_{SE}(t)=\rho_{S}(0)\otimes\rho_{E}-\ii\int_{0}^{t}ds\left[\bar{H}_{I}(s),\bar{\rho}_{SE}(s)\right].
\label{eq:vN_integral}
\EE

Tracing out the environment in \eqnref{eq:vN}, so that its l.h.s.~reduces to
$d\bar{\rho}_{S}(t)/dt$ and substituting into its r.h.s.~for $\bar{\rho}_{SE}(t)$ according to \eqnref{eq:vN_integral}, one arrives at the integro-differential equation
describing the system density matrix in the IP at time $t$: 
\BEA
\frac{d\bar{\rho}_{S}(t)}{dt}&=&-\ii\tr_{E} \left[\bar{H}_{I}(t),\rho_{S}(0)\otimes\rho_{E}\right] \label{eq:QME_exact_init}\\
&&-\int_{0}^{t}ds\,\tr_{E} \left[\bar{H}_{I}(t),\left[\bar{H}_{I}(s),\bar{\rho}_{SE}(s)\right]\right].\nonumber
\EEA
The first term in \eqnref{eq:QME_exact_init} may be dropped, as without loss of generality one may impose
\BE
\tr_{E}\!\left\{ \bar{H}_{I}(t)\rho_{E}\right\} = 0.
\label{eq:first_mom_vanish}
\EE
by shifting the zero point energy of Hamiltonians, i.e., by changing $H_{I}$ and $H_{S}$ as follows
\BE
H'_{I}=H{}_{I}-\tr_{E}\!\left\{ H_{I}\rho_{E}\right\} \otimes\openone_{E},\; 
H'_{S}=H_{S}+\tr_{E}\!\left\{ H_{I}\rho_{E}\right\},
\label{eq:HI_shift}
\EE
so that condition \eref{eq:first_mom_vanish} is ensured, given $\left[H_{E},\rho_{E}\right]=0$, without affecting the total Hamiltonian $H_{T}$ in \eqnref{eq:H_T}.

As a result, we obtain the QME in its \emph{integro-differential} form that does \emph{not} involve any approximations, but only assumes \eqnref{eq:init_prod} with $\left[H_{E},\rho_{E}\right]=0$,
\BE
\frac{d\bar{\rho}_{S}(t)}{dt}=-\int_{0}^{t}ds\,\tr_{E}\!\left\{ \left[\bar{H}_{I}(t),\left[\bar{H}_{I}(s),\bar{\rho}_{SE}(s)\right]\right]\right\},
\label{eq:QME_exact}
\EE
and constitutes \eqnref{eq:exact_integrodiff_singlebath} of the main text.

\subsection{QME in the time-local form}
\label{app:QME_TL}
The QME \eref{eq:QME_exact} despite being compact and exact is typically not of much use, as it involves the full system-environment state and a time-convoluted integral. Nevertheless, one may always formally rewrite it as a function of the system state at a given time. 

After integrating the closed von Neumann dynamics \eref{eq:vN}, one should arrive at
\BE
\bar{\rho}_{SE}(t)=\bar{U}_{SE}(t)\left(\rho_{S}(0)\otimes\rho_{E}\right)\bar{U}_{SE}^{\dagger}(t)
\label{eq:bar_rho_SE(t)}
\EE
with the unitary rotation being formally defined as a time-ordered
exponential:
\BE
\bar{U}_{SE}(t):=\mathcal{T}_{\leftarrow}\exp\!\left\{ -\ii\intop_{0}^{t}\!ds\;\bar{H}_{I}(s)\right\}.
\label{eq:U_SE_IP}
\EE

Now, as the reduced state of the system is obtained at any time by tracing out the environment, 
the \emph{dynamical map}, $\bar{\Lambda}_{t}$, associated solely with the system evolution in the IP may be identified as
\BE
\bar{\rho}_{S}(t)=\bar{\Lambda}_{t}\!\left[\rho_{S}(0)\right]:=\tr_E\!\left\{\bar{U}_{SE}(t)\left(\rho_{S}(0)\otimes\rho_{E}\right)\bar{U}_{SE}^{\dagger}(t)\right\}.
\label{eq:map_IP}
\EE 

Hence, if at a given $t$ one can compute the inverse of the dynamical map, i.e., $\bar{\Lambda}_{t}^{-1}$ such that $\rho_{S}(0)=\bar{\Lambda}_{t}^{-1}\!\left[\bar{\rho}_{S}(t)\right]$, as well as its the time-differential $\dot{\bar{\Lambda}}_{t}$, which is now formally determined by \eqnref{eq:QME_exact} as
\begin{align}
&\dot{\bar{\Lambda}}_{t}\left[\bullet\right]= 
\label{eq:dyn_map_IP_macro}\\
&
-\int_{0}^{t}\!\!ds\,\tr_{E}\left[\bar{H}_{I}(t),\left[\bar{H}_{I}(s),\bar{U}_{SE}(s)\left(\bullet\otimes\rho_{E}\right)\bar{U}_{SE}^{\dagger}(s)\right]\right] \nonumber,
\end{align}
one may equivalently rewrite the exact QME \eref{eq:QME_exact} into its \emph{time-local} form (in the IP):
\BE
\frac{d\bar{\rho}_{S}(t)}{dt}=\LIP_t\left[\bar{\rho}_{S}(t)\right],
\label{eq:QME_TL_IP}
\EE
where $\LIP_t=\dot{\bar{\Lambda}}_{t}\circ\bar{\Lambda}_{t}^{-1}$ is the IP-based dynamical generator defined in \eqnref{eq:L_IP}.

\subsection{QME in the Schr\"odinger picture and the environment-induced generator}
\label{app:QME_SP}

We rewrite the QME \eref{eq:QME_TL_IP} in the SP as
\BE
\frac{d\rho_{S}(t)}{dt} = \LSP_{t}\left[\rho_{S}(t)\right],
\label{eq:QME_SP}
\EE
where in accordance with \eqnref{eq:L_t_IP_SP} the IP-based dynamical generator must be transformed to
\BE
\LSP_t[\bullet] = -\ii[H_S,\bullet]+\cU_t^{S}\circ\LIP_t\circ\cU_t^{S\dagger}[\bullet].
\label{eq:chocalito}
\EE
As a result, we arrive at the time-local QME as stated in \eqnref{eq:QME} of the main text:
\BE
\frac{d\rho_{S}(t)}{dt} = -\ii\left[H_{S},\rho_{S}(t)\right]+\Lenv_{t}\left[\rho_{S}(t)\right],
\label{eq:QME_main}
\EE
and identify with help of \eqnref{eq:chocalito}:
\BE
\Lenv_{t}
:=
\LSP_t - \dot{\cU}_t^{S}\circ\cU_t^{S\dagger}
=
\cU_t^{S}\circ\LIP_t\circ\cU_t^{S\dagger}
\label{eq:L_env_ind}
\EE
as the \emph{environment-induced dynamical generator}, which can be then associated solely with the impact of the environment on the system.

Still, as $\Lenv_t$ generally contains both Hamiltonian and dissipative parts, i.e., $\Lenv_{t}\left[\bullet\right]=-\ii\left[H(t),\bullet\right]+\cD_{t}\left[\bullet\right]$, one may conveniently rewrite the QME \eref{eq:QME_SP} as \cite{Breuer,Rivas2012}:
\BE
\frac{d\rho_{S}(t)}{dt}=-\ii \left[H_{S}+H(t),\rho_{S}(t)\right]+\cD_{t}\left[\rho_{S}(t)\right],
\label{eq:QME_main_can}
\EE
which allows to explicitly identify $H(t)$ as the environment-induced Hamiltonian correction to the system free evolution, e.g., representing the Lamb shifts when describing atom-light interactions \cite{Breuer}. Whereas, $\cD_{t}$ in \eqnref{eq:QME_main_can} may then be entirely associated with the dissipative impact of the environment.

Lastly, let us emphasise that when investigating whether by \emph{adding generators} associated with different environments---i.e., the generators $\Lenv_t$ in \eqnref{eq:QME_SP} obtained by considering the impact of each environment separately---one reproduces the correct dynamics, it is equivalent to consider all the generators in the IP. 

As the IP-based generators, $\LIP_t$, are linearly related to the environment-induced ones, $\Lenv_t$, by \eqnref{eq:L_env_ind}, the vector spaces formed by their families must be isomorphic. Hence, the notions of \emph{rescalability} and \emph{additivity} of generator families, discussed in \secref{sub:gen_add} of the main text and \appref{app:QMEs_dyn_gen_fams} above, are naturally carried over between the two, e.g., for any $\alpha,\beta,t\ge0$:
\BE
\Lenv[\prime]_t=\alpha\Lenv[(1)]_t+\beta\Lenv[(2)]_t
\quad\Longleftrightarrow\quad
\LIP[\prime]_t=\alpha\LIP[(1)]_t+\beta\LIP[(1)]_t
\EE
with \eqnref{eq:L_env_ind} relating all $\Lenv[\msf{x}]_t=\cU_t^{S}\circ\LIP[\msf{x}]_t\circ\cU_t^{S\dagger}$ for each $\msf{x}=\{\prime,(1),(2)\}$.

\subsection{$H_S$-covariant dynamics}
\label{app:HS_cov}
Although the dynamical generators $\Lenv_t$ defined in \eqnref{eq:L_env_ind} arise due to the presence of the environment, their form may still strongly depend on the system free Hamiltonian $H_{S}$. Thus, $\Lenv_t$ and, in particular, both its  Hamiltonian $H(t)$ and dissipative parts $\cD_t$ in \eqnref{eq:QME_main_can} \emph{cannot} be generally associated with the properties of just the environment and the interactions. In fact, only in very special cases the form of $\cL_t$ can be derived from $H_{I}$, $H_{E}$, and $\rho_{E}$.

An important example is provided when the system and interaction Hamiltonians alone commute: 
\BE
\left[H_{\mathrm{S}},H_{\mathrm{I}}\right]=0.
\label{eq:HS_cov_comm}
\EE
As the global unitary $\bar{U}_{SE}$ in \eqnref{eq:U_SE_IP} then also commutes with $H_S$, $\left[\bar{U}_{SE},H_{S}\right]=0$, the IP-based map $\bar{\Lambda}_{t}$ in \eqnref{eq:map_IP} is assured to be \emph{$H_S$-covariant}, i.e., to commute with any $H_S$-induced unitary \eref{eq:U_S_map} (and so must trivially the SP-based $\Lambda_t$), so that for any $s\ge 0$ \citep{Holevo1993}:
\BE
\cU_s^S \circ \bar{\Lambda}_{t} = \bar{\Lambda}_{t}\circ \cU_s^S 
\quad\Longleftrightarrow\quad
\cU_s^S \circ \LIP_{t} = \LIP_{t} \circ \cU_s^S.
\label{eq:H_S_cov_Lambda}
\EE
As noted above, the $H_S$-covariance must be naturally inherited by the IP-based dynamical generators $\LIP_{t}$ \citep{Holevo1993,Vacchini2010}, which, in turn, must then coincide with the environment-induced ones, with $\Lenv_{t}=\LIP_{t}$ in \eqnref{eq:L_env_ind}.
As a result, the form of $\Lenv_{t}$ must then, indeed, be independent of $H_S$.

\subsection{Externally modifying the system Hamiltonian }
\label{app:QME_HS}
In general, a modification of the system free Hamiltonian $H_S$ may affect both the Hamiltonian and the dissipative parts of the generator $\Lenv_t$ in \eqnref{eq:QME_main_can}. Nevertheless, let us consider a transformation:
\BE
H_S \to H^{\prime}_S(t):=H_S+V(t)
\EE 
with $V(t)$ being an arbitrary (potentially time-dependent) Hermitian operator. Crucially, by considering particular commutation relations satisfied by the microscopic Hamiltonians of \eqnref{eq:H_T} and $V(t)$, one may identify two important cases for which the microscopic rederivation of the QME \eref{eq:QME_main_can} can be bypassed---with the impact of $V(t)$ being directly accountable for at the level of the QME:

\paragraph{$\left[H_{S},H_{I}\right]=0$, $\forall_{t\ge0}\!\!:\left[V(t),H_{I}\right]=0$.}
If the system Hamiltonian commutes with the interaction Hamiltonian---so that the dynamics is $H_S$-covariant---and so does the perturbation $V(t)$ for all $t$, then the modified dynamics must also be $H^{\prime}_S$-covariant, as $\left[H^{\prime}_{S}(t),H_{I}\right]=0$ at all times. Hence, the form of $\Lenv_t$ in \eqnref{eq:QME_main} is unaffected by the modification of $H_S$, remaining fully determined by $H_E$, $H_I$ and $\rhoE$. Moreover, the new dynamics is then correctly described by simply replacing $H_S$ with $H^{\prime}_S(t)$ in \eqnref{eq:QME_main} (or \eqnref{eq:QME_main_can}).

\paragraph{$\left[H_{S},H_{I}\right] \neq 0$, $\forall_{t\ge0}\!\!:\left[V(t),H_{I}\right]=\left[V(t),H_{S}\right]=0$.}
The above conclusion also holds when dealing with non-$H_S$-covariant dynamics, given that $V(t)$ commutes with both the interaction and the system Hamiltonian. As then $\left[H_{S},H_{I}\right] \neq 0$, the generator $\Lenv_t$ in principle depends on $H_S$. However, without affecting the total Hamiltonian $H_{T}$ in \eqnref{eq:H_T} and hence the dynamics, we may redefine the interaction Hamiltonian as $H_I':=H_I+H_S$, pretending the system Hamiltonian to be absent. In such a fictitious picture, the QME \eref{eq:QME_main} possesses just the second term with $\Lenv[\prime]_t$ now being derived based on $H_I'$. As importantly $\left[V(t),{H}_I'\right]=0$ is fulfilled at all times, it becomes clear that the dynamics must be $V(t)$-covariant. Hence, the perturbation must lead to a QME that may be equivalently obtained by simply adding $V(t)$ to the Hamiltonians in \eqnref{eq:QME_main_can}---even though $H(t)$ and $\cD_t$ non-trivially depend on the original $H_S$ (but not on $V(t)$).

Lastly, let us note that in case of $H_S$-covariant dynamics and $\left[H_{S},H_{I}\right]=0$, one may play a similar trick in order to deal with the case when $\left[V(t),H_{S}\right]=0$ but $\left[V(t),H_{I}\right]\neq 0$, so that the modified dynamics is no longer guaranteed to be $H^{\prime}_S$-covariant. By redefining the interaction Hamiltonian this time as $H_{I}':=V(t)+H_{I}$, which importantly commutes with $H_{S}$, it becomes clear that the $H_S$-covariance must be preserved. Nonetheless, although the dynamical generator $\Lenv_t$ in \eqnref{eq:QME_main} remains then independent of $H_S$, the form of $\Lenv_t$ may depend on $V(t)$ and must thus be rederived, i.e., based now on $H_I'$.

\section{Microscopic validity of generator addition}
\label{app:valid_add_gens}

\subsection{Weak-coupling regime}
\label{app:valid_add_gens_weakcoupling}
Below, we prove \lemref{lem:weak_coupling} stated in the main text, in particular, we show that under weak coupling the cross-term in \eqnref{eq:QME_multiple_envs} can always be assumed to vanish. Hence, in accordance with \obsref{obs:crossterm_vanish}, it is then valid to add dynamical generators corresponding to each individual environment, without necessity to rederive the overall QME. 

The following proof can be regarded as an extension of the argumentation found in \refcite{CohenTannoudji1998,Schaller2015}, which applies to the more stringent regime in which the Born-Markov approximation holds.

Firstly, we perform the operator Schmidt decomposition of each interaction Hamiltonian (indexed by $i$) \cite{Bengtsson2006}:
\BE
H_{I_i} \!=\! \sum_{k} A_{i;k} \otimes B_{k}^{E_i} 
\;\;\Leftrightarrow\;\;
\bar{H}_{I_i}(t) \!=\! \sum_{k} \bar{A}_{i;k}(t) \otimes \bar{B}_{k}^{E_i}(t).
\label{eq.HImultidecomp}
\EE
where $\{A_{i;k}\}_k$ and $\{B_k^{E_i}\}_k$ form then sets of Hermitian operators that act separately on the system and corresponding environment subspaces, i.e., $\cH_S$ and $\cH_{E_i}$, respectively. 

As noted above, this decomposition preserves its tensor-product structure in the IP, which is now defined according to \eqnref{eq:IP_def} with the free system-environment Hamiltonian incorporating multiple environments, $H_S + \sum_i H_{E_i}$. Hence, carrying out here the analysis in the IP for compactness of the expressions, we rewrite the general and exact QME \eref{eq:QME_multiple_envs} of the main text, which describes a system interacting with multiple environments, as
\begin{widetext}
\BEA
\frac{\dd}{\dd t}{\bar{\rho}}_S(t) 
& = & \sum_{i}\LIP[(i)]_t[{\bar{\rho}}_S(t)]+\sum_{i\ne j} \int_0^t \dd s \tr_{E_{ij}}\!\left\{  \left[\bar{H}_{I_i}(t), \left[ \bar{H}_{I_j}(s), \bar{\rho}_{SE_{ij}}(s) \right]\right]\right\} \label{eq:QME_IP_multiple}\\
& = & 
\sum_{i}\LIP[(i)]_t[{\bar{\rho}}_S(t)]+
\sum_{i\ne j} \sum_{k,l} \int_0^t ds\; \tr_{E_{ij}}\!\left\{  \left[\bar{A}_{i;k}(t)\otimes \bar{B}_{k}^{E_i}(t), \left[\bar{A}_{j; l}(s)\otimes \bar{B}_{l}^{E_j}(s), \bar{\rho}_{SE_{ij}}(s)\right]\right]\right\},
\label{eq:QME_ABdecompS}
\EEA
\end{widetext}
where the second term above is the crucial, inter-environment cross-term whose absence assures the validity of generator addition.

There exist various approaches to obtain simplified forms of QMEs for the weak-coupling regime \cite{deVega2017,Breuer,Schaller2015,Rivas2012}. Here, in order keep the derivation general and emphasise necessary requirements for our arguments to apply, we assume that the appropriate QME under weak coupling is derived after approximating the global system-environments state at every time $t\ge0$ as
\BE
\rho_{SE}(t) \approx  \rho_{S}(t) \otimes  \bigotimes_i \varrho_{E_i}(t),
\label{eq:sep_sys_envs_approx}
\EE
where $\rho_{S}(t)=\tr_{E}\rho_{SE}(t)$ is the reduced state of the system at time $t$. Although in the weak-coupling approximations \cite{Breuer,Rivas2012} the separable state of each environment in \eqnref{eq:sep_sys_envs_approx} is frequently taken to be its reduced state at time $t$, i.e., $\varrho_{E_i}(t)\equiv\rho_{E_i}(t):=\tr_{\neg E_i}\rho_{SE}(t)$, in what follows it can be chosen arbitrarily---as long as for all environments (labelled by $i$) $\varrho_{E_i}(0)=\rho_{E_i}$ to maintain consistency with the derivation in \appref{app:QME_integrodiff}.

Let us stress that the tensor-product ansatz of \eqnref{eq:sep_sys_envs_approx} for the system-environments state is employed only to obtain the form of the QME valid under the weak coupling, and does \emph{not} force the solutions of the QME to actually be separable states. In particular, the resulting QME, before tracing out environmental degrees of freedom, can yield upon integration states $\bar{\rho}_{SE}(t)$ that strongly deviate from the form \eref{eq:sep_sys_envs_approx} already at moderate times $t$, even though the validity of the dynamics---and, hence, the QME employed---is still assured by weak coupling \cite{Rivas2010}.

Thanks to the condition \eref{eq:sep_sys_envs_approx}, the crucial cross-term within the exact QME \eref{eq:QME_ABdecompS} can be reexpressed as 
\begin{align}
&
\sum_{i\ne j} \sum_{k,l} \int_0^t ds\quad \Cfun_{[i;k][j;l]}(t,s;s)\times 
\label{eq:cross_term_Cikjl}\\
& \quad\times\left(\bar{A}_{i;k}(t)\bar{A}_{j;l}(s)\bar{\rho}_{S}(s)-\bar{A}_{i;k}(t)\bar{\rho}_{S}(s)\bar{A}_{j;l}(s) \right.\nonumber\\
& \qquad\quad\left.-\bar{A}_{j;l}(s)\bar{\rho}_{S}(s)\bar{A}_{i;k}(t)+\bar{\rho}_{S}(s)\bar{A}_{j;l}(s)\bar{A}_{i;k}(t)\right), 
\nonumber
\end{align}
where now
\begin{align}
\Cfun_{[i;k][j;l]}(t,s;s') & 
:=
\tr\!\left\{ \left(\bar{B}_{k}^{E_{i}}(t)\otimes\bar{B}_{l}^{E_{j}}(s)\right)\bar{\varrho}_{E_{ij}}(s')\right\},
\label{eq:C_ikjl}
\end{align}
is the \emph{two-bath correlation function} that is independent of the reduced system state, being evaluated only on $\bar{\varrho}_{E_{ij}}(t):=\bar{\varrho}_{E_{i}}(t)\otimes\bar{\varrho}_{E_{j}}(t)$. Note that, as $i$ and $j$ are just labels of distinct environments, the correlation function \eref{eq:C_ikjl} is symmetric with $\Cfun_{[i;k][j;l]}(t,s;s')=\Cfun_{[j;l][i;k]}(s,t;s')$.

Moreover, \eqnref{eq:sep_sys_envs_approx} assures all the correlation functions \eref{eq:C_ikjl} to factorise, so that for any $s,s',t\ge0$:
\begin{align}
\Cfun_{[i;k][j;l]}(t,s;s') & \approx\tr\!\left\{ \bar{B}_{k}^{E_{i}}(t)\bar{\varrho}_{E_{i}}(s')\otimes\bar{B}_{l}^{E_{j}}(s)\bar{\varrho}_{E_{j}}(s')\right\} \\
 & =\Cfun_{i;k}(t,s')\;\Cfun_{j;l}(s,s'),
\label{eq:C_ikjl_weak}
\end{align}
reducing to a product of \emph{single-bath correlation functions}:
\BE
\Cfun_{i;k}(t,s):=\tr\!\left\{ \ee^{\ii H_{E_{i}}(t-s)}B_{k}^{E_{i}}\ee^{-\ii H_{E_{i}}(t-s)}\varrho_{E_{i}}(s)\right\}.
\label{eq:_Cik_ts}
\EE
Furthermore, as the two-bath correlation function in \eqnref{eq:cross_term_Cikjl} factorises to \eqnref{eq:C_ikjl_weak} with $s=s'$, the whole cross-term \eref{eq:cross_term_Cikjl} is guaranteed to vanish whenever at all times $t\ge0$ for each environment (labelled by $i$) and its each operator $B_{k}^{E_{i}}$ (labelled by $k$):
\BE
\Cfun_{i;k}(t):=\Cfun_{i;k}(t,t)=\tr\{B_{k}^{E_{i}}\,\varrho_{E_{i}}(t)\}\underset{?}{=}0.
\label{eq:_Cik_tt}
\EE

Crucially, the condition \eqnref{eq:_Cik_tt} can always be ensured by shifting adequately the interaction and the system Hamiltonians without affecting the total Hamiltonian \eref{eq:H_T} (similarly to \eqnref{eq:HI_shift} of \appref{app:QME_integrodiff}). In particular, one can redefine the system and each interaction Hamiltonian to be generally time-dependent and read:
\begin{align}
H'_{S}(t) &= H_{S}+\sum_{i,k}\Cfun_{i;k}(t)\,A_{i;k},\\
\forall_i:\;H'_{I_i}(t) &= H_{I_i}-\sum_k\Cfun_{i;k}(t)\,(A_{i;k}\otimes\1_{E_i}),
\end{align}
so that the decomposition \eref{eq.HImultidecomp} of the interaction Hamiltonian for each environment becomes
\BE
H_{I_i}'(t) = \sum_{k} A_{i;k} \otimes B_{k}^{\prime E_i}(t)= \sum_{k} A_{i;k} \otimes [B_{k}^{E_i}-\Cfun_{i;k}(t)\1_{E_i}],
\label{eq:HIprime_decomp}
\EE
with the new correlation function \eref{eq:_Cik_tt} identically vanishing by construction, as for any $t\ge 0$:
\begin{align}
\Cfun_{i;k}^\prime(t)
&=
\tr_{E_{i}}\!\left\{B_{k}^{\prime E_i} \varrho_{E_{i}}(t)\right\}
\label{eq:puigdemond_pipa}\\
&=
\tr_{E_{i}}\!\left\{(B_{k}^{E_i}-\Cfun_{i;k}(t))\varrho_{E_{i}}(t)\right\}=0. \nonumber
\end{align}

For consistency, le us also note that the necessary requirement $\tr_{E}\!\left\{ \bar{H}_{I_i}^\prime(t)\rho_{E_i}\right\}=0$, introduced in \appref{app:QME_integrodiff}, is then trivially fulfilled for every environment. Decomposing $H'_{I_i}(t)$ according to \eqnref{eq:HIprime_decomp} and remembering that $[H_{E_i},\rho_{E_i}]=0$ for each $i$, one gets (in the IP):
\BE
\sum_{k}\bar{A}_{i;k}(t)\tr_{E_{i}}\!\left\{B_{k}^{\prime E_{i}}\rho_{E_{i}}\right\} =\sum_{k}\Cfun_{i;k}(0)\,\bar{A}_{i;k}(t) = 0,
\EE
as each $\Cfun_{i;k}^\prime(0)=0$ is zero by \eqnref{eq:puigdemond_pipa}.

Note that, in particular, the above argumentation holds for all QMEs derived using the \emph{time-convolutionless} approach \cite{Breuer2001} up to the \emph{second order} in all the interaction parameters---in which case the QME \eref{eq:QME_IP_multiple} is from the start assumed to exhibit a time-local form, rather than involve a time-convolution integral. 

On the other hand, the most conservative \emph{Born-Markov approximation} discussed in \refcite{CohenTannoudji1998,Schaller2015} enforces every $\varrho_{E_i}(t)$ in \eqnref{eq:sep_sys_envs_approx} to be at all times the initial, stationary state $\rho_{E_i}$ of each environment. As a result, all the single-bath correlation functions, $\Cfun_{i;k}(t,s)$ in \eqnref{eq:_Cik_ts}, become then $t$- and $s$-independent due to $[H_{E_i},\rho_{E_i}]=0$, and identically vanish by \eqnref{eq:puigdemond_pipa}. Hence, then trivially $\Cfun_{[i; k][j;l]}=\Cfun_{i;k}\Cfun_{j;l}=0$ at all times.

\subsection{Commutativity of microscopic Hamiltonians}
\label{app:valid_add_gens_comm}
Here, we provide the proof of \lemref{lem:commutativity} stated in the main text, which assures that dynamical generators associated with each individual environment can be simply added at the QME level, if the interaction Hamiltonians commute between each other and with the system Hamiltonian. This condition corresponds to the $\text{II}\cap\text{IS}$ region in the Venn diagram of \figref{fig:venn}---marked `Yes' to indicate the validity of generator addition.

Let us note that whenever for all $i$ and $j$:
\BE
\left[H_{\mathrm{I}_{i}},H_{\mathrm{I}_{j}}\right]=0 \qquad\trm{and}\qquad \left[H_{I},H_{S}\right]=0,
\label{eq:ass_IIandIS}
\EE
with $H_{I}:=\sum_i H_{I_i}$ being the full interaction Hamiltonian, the global unitary dynamical operator \eref{eq:U_SE_IP} can be decomposed in the SP, as follows
\BE
U_{SE}(t) = \ee^{-\ii (H_S + H_E + H_I) t} = \ee^{-\ii H_S t } \prod_i \ee^{-\ii (H_{E_i} + H_{I_i}) t},
\label{eq.totaluni}
\EE
so that at the level of the corresponding unitary maps:
\BE
\cU_t^{SE} = \cU_t^S \circ \prod_i \cU_t^{IE_i},
\label{eq.totaluni_maps}
\EE
where we have defined $\cU_t^{IE_i}[\bullet]:=\ee^{-\ii (H_{E_i} + H_{I_i})t}\bullet\ee^{\ii (H_{E_i} + H_{I_i})t}$, and by $\prod$ we denote also the conjugation of multiple maps, i.e., for a given set of maps $\{\Lambda_i\}_i$: 
\BE
\prod_{i=1}^n\Lambda_i
:=
\Lambda_n\circ\Lambda_{n-1}\circ\dots\circ\Lambda_2\circ\Lambda_1.
\EE

As a result, after straightforwardly generalising \eqnref{eq:map_IP} to multiple environments and transforming it to the SP, we can generally write the system reduced state at a given time $t$ as
\begin{widetext}
\begin{align}
\rhoS(t) & = \tr_E \!\left\{\cU_t^{SE} [\rhoS(0) \ot \rhoE]  \right\} 
= \tr_E \!\left\{\,\cU_t^{S} \circ \prod_i \cU_t^{SE_i}\!\left[\rhoS(0) \ot \,\Ot_k \rho_{E_k}\right]  \right\}\\
& = \cU_t^{S} \!\left[\,\tr_{E_{i\neq 1}} \!\left\{\, \prod_{i\neq 1} \cU_t^{SE_i}\! \left[\, \tr_{E_{1}}\! \left\{\cU_t^{SE_1} \!\left[\rhoS(0)\ot \rho_{E_1}\right]\right\}\ot\,  \Ot_{k \neq 1} \rho_{E_k} \right] \!\right\} \!\right] \\
& = \cU_t^{S} \!\left[\, \tr_{E_{i\neq 1}}\!\left\{\, \prod_{i\neq 1} \cU_t^{SE_i} \!\left[\, \LambdaE_t^{(1)}[\rhoS(0)] \ot\, \Ot_{k \neq 1} \rho_{E_k} \right]\!\right\}\right] = \cdots = \label{dupa1}\\
& = \cU_t^{S}\!\left[\, \tr_{E_{i\neq 1,2}}\!\left\{\, \prod_{i\neq 1,2} \cU_t^{SE_i}\!\left[\, \LambdaE_t^{(2)}\circ\LambdaE_t^{(1)}[\rhoS(0)] \ot\, \Ot_{k \neq 1,2} \rho_{E_k} \right] \right\} \right] = \cdots = \label{dupa2}\\
& = \cU_t^S \!\left[\, \prod_i \LambdaE_t^{(i)} [\rhoS(0)] \right] 
\quad=:\quad 
\cU_t^S \circ \LambdaE_t [\rhoS(0)],
\label{eq:dupablada}
\end{align}
\end{widetext}
where in \eref{dupa1} by $\cdots$ we mean repeating the procedure for the $2$nd environment, and similarly  in \eref{dupa2} for all the other environments. The overall dynamical (SP-based) map $\LambdaE_t=\prod_i \LambdaE_t^{(i)}$ is $H_S$-independent and constitutes a composition of maps $\LambdaE_t^{(i)}[\bullet] := \tr_{E_{i}} \!\left\{\cU_t^{SE_i} [\bullet \ot \rho_{E_i}] \right\}$, each of which describing the impact of the $i$th environment.

As discussed in \appref{app:HS_cov} above, the condition $[H_{I},H_{S}]=0$ in \eqnref{eq:ass_IIandIS} ensures the dynamics to be $H_S$-covariant---as explicitly manifested in \eqnref{eq:dupablada} in which $\cU_t^S$, thanks to commuting with all $\cU_t^{SE_i}$, commutes also with all the $\LambdaE_t^{(i)}$ maps. On the other hand, as all $\cU_t^{SE_i}$ commute between one another due to $[H_{I_i},H_{I_j}]=0$ in \eqnref{eq:ass_IIandIS}, the overall map $\LambdaE_t$ in \eqnref{eq:dupablada} could have been constructed by composing the maps  $\LambdaE_t^{(i)}$ in any order. Hence, all maps originating from interactions with different reservoirs commute, i.e., for all $i$ and $j$:
\BE
\LambdaE_t^{(i)}\circ\LambdaE_t^{(j)}=\LambdaE_t^{(j)}\circ\LambdaE_t^{(i)}.
\label{eq:comm_channels}
\EE

As a consequence, while due to the $H_S$-covariance all the dynamical generators induced by separate environments, i.e., $\Lenv[(i)]_t$ in \eqnref{eq:L_env_ind} indexed now by $i$, coincide with their corresponding IP-based generators, they also add at the level of the QME \eref{eq:QME_main} thanks to \eqnref{eq:comm_channels}. Computing explicitly the dynamical generator induced by all the environments together, i.e., $\Lenv_t$ associated with the overall map $\LambdaE_t$ in \eqref{eq:dupablada}, we have
\BEA
\Lenv_t = \dot{\LambdaE}_t \circ \LambdaE_t^{-1} 
&=& 
\left( \sum_i \dot{\LambdaE}_t^{(i)} \circ \prod_{j\neq i} \LambdaE_t^{(j)} \right) \circ \prod_k \left(\LambdaE_t^{(k)}\right)^{-1} \nonumber\\
&=& \sum_i \dot{\LambdaE}_t^{(i)}\circ \left(\LambdaE_t^{(i)}\right)^{-1}= \sum_i \Lenv[(i)]_t,
\EEA
where $\Lenv[(i)]_t$ is the generator corresponding to the interaction with the $i$th environment alone, and we have used the commutativity of the maps \eref{eq:comm_channels}.

\section{Spin-magnet model}
\label{app:spin-magnet}
We provide here the details and explicit form of relevant quantities for the calculations presented in \secref{subsec:Multiple} of the main text, discussing the counterexamples to the commutativity assumptions based on the spin-magnet model.

\subsection{$\text{IS}\cap\text{IE}$ commutativity assumption}
\label{app:IS_IE}
We solve the equations of motion \eref{firsteqsmotion} that describe the Bloch vector dynamics in order to obtain an explicit form of the $\mat{R}$-matrix in \eqnref{eq:affine_map}, which is then parametrised by magnetisations of the two magnets, $m_1$ and $m_2$, and reads:
\begin{widetext}
\BE
\mat{R}_{12}(m_1,m_2,t) = 
\begin{pmatrix}
 \frac{\cos \left(t \sqrt{g_1^2 m_1^2+g_2^2 m_2^2}\right) g_1^2 m_1^2+g_2^2 m_2^2}{g_1^2 m_1^2+g_2^2 m_2^2} & -\frac{\sin \left(t \sqrt{g_1^2 m_1^2+g_2^2 m_2^2}\right) g_1 m_1}{\sqrt{g_1^2 m_1^2+g_2^2 m_2^2}} & \frac{2 \sin ^2\left(\frac{1}{2} t \sqrt{g_1^2 m_1^2+g_2^2 m_2^2}\right) g_1 g_2 m_1 m_2}{g_1^2 m_1^2+g_2^2 m_2^2} \\
 \frac{\sin \left(t \sqrt{g_1^2 m_1^2+g_2^2 m_2^2}\right) g_1 m_1}{\sqrt{g_1^2 m_1^2+g_2^2 m_2^2}} & \cos \left(t \sqrt{g_1^2 m_1^2+g_2^2 m_2^2}\right) & -\frac{\sin \left(t \sqrt{g_1^2 m_1^2+g_2^2 m_2^2}\right) g_2 m_2}{\sqrt{g_1^2 m_1^2+g_2^2 m_2^2}} \\
 \frac{2 \sin ^2\left(\frac{1}{2} t \sqrt{g_1^2 m_1^2+g_2^2 m_2^2}\right) g_1 g_2 m_1 m_2}{g_1^2 m_1^2+g_2^2 m_2^2} & \frac{\sin \left(t \sqrt{g_1^2 m_1^2+g_2^2 m_2^2}\right) g_2 m_2}{\sqrt{g_1^2 m_1^2+g_2^2 m_2^2}} & \frac{g_1^2 m_1^2+\cos \left(t \sqrt{g_1^2 m_1^2+g_2^2 m_2^2}\right) g_2^2 m_2^2}{g_1^2 m_1^2+g_2^2 m_2^2} \\
\end{pmatrix}.
\label{eq:R_12_case1}
\EE
\end{widetext}
When only the first magnet is present ($g_2=0$), the above expression reduces to
\BE
\label{eq.nointsolmag1}
\mat{R}_1(m_1,t) =
\begin{pmatrix}
 \cos \left(t g_1 m_1\right) & -\sin \left(t g_1 m_1\right) & 0 \\
 \sin \left(t g_1 m_1\right) & \cos \left(t g_1 m_1\right) & 0 \\
 0 & 0 & 1 \\
\end{pmatrix},
\EE
while, when in contact with only the second magnet ($g_1=0$), it becomes
\BE
\mat{R}_2(m_2,t) = 
\begin{pmatrix}
 1 & 0 & 0 \\
 0 & \cos \left(t g_2 m_2\right) & -\sin \left(t g_2 m_2\right) \\
 0 & \sin \left(t g_2 m_2\right) & \cos \left(t g_2 m_2\right) \\
\end{pmatrix}.
\label{eq:kaczynski_cipa}
\EE
We also analytically compute the time-derivative $\dot{\mat{R}}_{12}$ (as well as $\dot{\mat{R}}_{1}$ and $\dot{\mat{R}}_{2}$, after setting $g_2=0$ and $g_1=0$, respectively) which we, however, do not include here due its cumbersome form.

With the exact expressions (\ref{eq:R_12_case1}-\ref{eq:kaczynski_cipa}) at hand, we can explicitly write each affine map $\mat{D}_t^{(\msf{x})}$ with $\msf{x}=\{12,1,2\}$ according to \eqnref{eq.numint} of the main text, i.e., as an average of the corresponding $\mat{R}_{\msf{x}}$ over the Gaussian distributions $p(m_i)$ of fixed variance $\sigma_i$ in \eqnref{eq.gaussianmagnet};~and similarly in case of the time-derivatives $\dot{\mat{D}}_t^{(\msf{x})}$ by averaging $\dot{\mat{R}}_{\msf{x}}$.

We perform the averaging integrals numerically after fixing the parameters $g_1$, $g_2$, $\sigma_1$, $\sigma_2$, and the time $t$. As a result, we obtain the expressions of dynamical generators $\LSPb[(1)]_t$, $\LSPb[(2)]_t$, and $\LSPb[(12)]_t$ by substituting into $\LSPb[(\msf{x})]_t= \dot{\mat{D}}_t^{(\msf{x})}(\mat{D}_t^{(\msf{x})})^{-1}$ the relevant affine maps and their time-derivatives.

For instance, when taking $g_1\!=\!g_2 = 2$ and $\sigma_1\!=\!\sigma_2 = 1$, we obtain at $t=0.5$:
\BE
\LSPb[(1)]_{t=0.5} = 
\begin{pmatrix}
 -2 & 0 & 0 \\
 0 & -2 & 0 \\
 0 & 0 & 0 \\
\end{pmatrix}, 
\quad
\LSPb[(2)]_{t=0.5} =
\begin{pmatrix}
 0 & 0 & 0 \\
 0 & -2 & 0 \\
 0 & 0 & -2 \\
\end{pmatrix},
\EE
and
\BE
\LSPb[(12)]_{t=0.5} = 
\begin{pmatrix}
 -1.56835 & 0 & 0 \\
 0 & -7.26687 & 0 \\
 0 & 0 & -1.56835 \\
\end{pmatrix},
\EE
which provides the desired example of $\LSPb[(12)]_t \neq \LSPb[(1)]_t + \LSPb[(2)]_t$.

\subsection{$\text{II}\cap \text{IE}$ commutativity assumption}
\label{app:II_IE}
We observe that, in order to solve the equations of motion \eref{secondeqsmotion} stated in the main text, which describe the dynamics of the Bloch vector in the IP, i.e., $\bar{\rV}(t)=\mat{R}_S^{-1}(t)\,\rV(t)$, it is convenient to move to a rotating frame defined as $\check{\rV}(t) := \mat{V}(t)\, \bar{\rV}(t)$, where
\BE
\mat{V}(t) 
:=
\begin{pmatrix}
 1 & 0 & 0 \\
 0 & -\cos(\omega t) & \sin(\omega t) \\
 0 & \sin(\omega t) & \cos(\omega t) \\
\end{pmatrix}
\EE
is an orthogonal matrix such that $\mat{V}(t)=\mat{P}\,\mat{R}_S(t)$ with $\mat{P}=\trm{diag}\{1,-1,1\}$ and  $\mat{R}_S(t)$ being the SO(3) representation of qubit unitary $U_S(t)=\ee^{-\ii H_S t}$ induced by the system free Hamiltonian $H_S=\frac{1}{2}\omega\sx$ of \eqnref{eq.magmodel2}. Hence, $\check{\rV}(t) = \mat{P}\, \rV(t)$ can be interpreted as the Bloch vector in the SP with the $y\to-y$ coordinate inverted.

Defining also $\gamma = g_1 m_1 + g_2 m_2$ for compactness, we obtain a simpler set of equations of motion:
\BE
\dot{\check{{r}}}_x = \gamma \check{r}_y,\quad
\dot{\check{{r}}}_y = \omega \check{r}_z - \gamma \check{r}_x,\quad
\dot{\check{{r}}}_z = - \omega \check{r}_y,
\EE
which can be explicitly solved, yielding 
\begin{widetext}
\BE
\check{\mat{R}}(m_1,m_2,t) =
\begin{pmatrix}
 \frac{\cos \left(\sqrt{\gamma ^2+\omega ^2}\,t\right) \gamma ^2+\omega ^2}{\gamma ^2+\omega ^2} & \frac{\gamma  \sin \left(\sqrt{\gamma ^2+\omega ^2}\,t\right)}{\sqrt{\gamma ^2+\omega ^2}} & -\frac{\gamma  \omega  \left(\cos \left(\sqrt{\gamma ^2+\omega ^2}\,t\right)-1\right)}{\gamma ^2+\omega ^2} \\
 -\frac{\gamma  \sin \left(\sqrt{\gamma ^2+\omega ^2}\,t\right)}{\sqrt{\gamma ^2+\omega ^2}} & \cos \left(\sqrt{\gamma ^2+\omega ^2}\,t\right) & \frac{\omega  \sin \left(\sqrt{\gamma ^2+\omega ^2}\,t\right)}{\sqrt{\gamma ^2+\omega ^2}} \\
 -\frac{\gamma  \omega  \left(\cos \left(\sqrt{\gamma ^2+\omega ^2}\,t\right)-1\right)}{\gamma ^2+\omega ^2} & -\frac{\omega  \sin \left(\sqrt{\gamma ^2+\omega ^2}\,t\right)}{\sqrt{\gamma ^2+\omega ^2}} & \frac{\gamma ^2+\omega ^2 \cos \left(\sqrt{\gamma ^2+\omega ^2}\,t\right)}{\gamma ^2+\omega ^2} \\
\end{pmatrix}.
\label{eq:rotatemyass}
\EE
\end{widetext}
We then construct the $\mat{R}$-matrix determining the affine map $\mat{D}_t$ in \eqnref{eq:affine_map} by transforming back the above $\check{\mat{R}}$-matrix to the IP, so that
\BE
\bar{\mat{R}}(m_1,m_2,t) = \mat{V}^{-1}(t) \, \check{\mat{R}}(m_1,m_2,t)\, \mat{V}(0).
\EE

We do not enclose here the explicit forms, but, as in the previous example, we also compute
all the relevant $\bar{\mat{R}}_\msf{x}$ and $\dot{\bar{\mat{R}}}_\msf{x}$ in the IP, with $\msf{x}=\{12,1,2\}$, which allow us to obtain the integral expressions for the corresponding affine maps, $\bar{\mat{D}}_t^{(\msf{x})}$, and their time-derivatives, $\dot{\bar{\mat{D}}}_t^{(\msf{x})}$. Again, we choose the initial magnetisations of both magnets to be Gaussian distributed with both $p(m_i)$ as in \eqnref{eq.gaussianmagnet} of fixed variance $\sigma_i$.

As before, we perform the averaging integrals numerically, after fixing the model parameters---now: $\omega$, $g_1$, $g_2$, $\sigma_1$, $\sigma_2$, and the time $t$---in order to obtain numerical expressions for  all $\LIPb[(\msf{x})]_t = \dot{\mat{D}}_t^{(\msf{x})}(\mat{D}_t^{(\msf{x})})^{-1}$. 

For example, when choosing $g_1 = g_2 = 2$, $\sigma_1=\sigma_2=1$ and $\omega = 2$, we obtain at  $t = 0.5$:
\BE
\LIPb[(1)]_{t=0.5} = \LIPb[(2)]_{t=0.5} = 
\begin{pmatrix}
 0.379798 & 0. & 0. \\
 0. & 0.779093 & -1.40007 \\
 0. & 1.70235 & -3.05922 \\
\end{pmatrix},
\EE
giving
\BE
\LIPb[(1)]_{t=0.5} + \LIPb[(2)]_{t=0.5} = 
\begin{pmatrix}
 0.759597 & 0. & 0. \\
 0. & 1.55819 & -2.80015 \\
 0. & 3.4047 & -6.11843 \\
\end{pmatrix}.
\EE
On the other hand, we find in presence of both magnets:
\BE
\LIPb[(12)]_{t=0.5} =
\begin{pmatrix}
 1.28248 & 0. & 0. \\
 0. & 13.0326 & -16.6865 \\
 0. & 28.4767 & -36.4608 \\
\end{pmatrix},
\EE
what, thus, provides an instance of $\LIPb[(12)]_{t} \neq \LIPb[(1)]_{t} + \LIPb[(2)]_{t}$. 

\end{document}